\documentclass[12pt,a4paper]{article}
\usepackage[utf8]{inputenc}
\usepackage[left=2cm,right=2cm,top=2cm,bottom=2cm]{geometry}

\usepackage{adjustbox}
\usepackage{enumitem}
\usepackage{chngcntr}
\usepackage[most]{tcolorbox}

\usepackage{graphicx}
\usepackage{multirow}
\usepackage{amsmath,amssymb,amsfonts}
\usepackage{amsthm}
\usepackage{mathrsfs}
\usepackage[title]{appendix}
\usepackage{xcolor}
\usepackage{textcomp}
\usepackage{manyfoot}
\usepackage{booktabs}
\usepackage{algorithmicx}
\usepackage{listings}

\usepackage{import} 
\usepackage{algorithm}
\usepackage{algpseudocode}
\usepackage{float}
\usepackage[section]{placeins}
\usepackage[compatibility=false]{caption}
\usepackage{subcaption}

\algnewcommand\algorithmicparameters{\textbf{Set:}}
\algnewcommand\SET{\Statex \algorithmicparameters\ } 
\algrenewcommand\algorithmicreturn{\textbf{return}}   
\newtheorem{proposition}{Proposition}

\makeatletter
\def\fps@figure{!htbp}
\def\fps@table{!htbp}
\makeatother
\graphicspath{ {./0_Figures/} }
\usepackage{authblk}
\usepackage[numbers, square]{natbib}
\usepackage{hyperref}

\title{ndscape: a structural map of the nested dichotomy space}

\author[1]{Maxwell Dix-Matthews\thanks{maxwell.dix-matthews@research.uwa.edu.au}}
\author[2]{Melinda Hodkiewicz}
\author[1]{Adriano Polpo}

\affil[1]{School of Physics, Mathematics and Computing, UWA, Australia}
\affil[2]{Department of Mechanical \& Chemical Engineering, UWA, Australia}

\date{}

\newcommand{\nd}[2]{\{(#1)\textnormal{ vs }(#2)\}}
\begin{document}

\maketitle

\begin{abstract}
Nested dichotomies (NDs) reduce multiclass classification to a tree of binary classifiers. Tree structure can change accuracy, bias, and variance, but the full space of structures has not been mapped in aggregate. We introduce ndscape, a two-dimensional mapping of ND structures based on depth-weighted structural tokens. We use this embedding to conduct what is, to our knowledge, the first bias-variance analysis across the ND structure space. In simulations with 6 to 10 classes, NDs with similar tree structures exhibit similar squared bias, variance, and squared probability loss, regardless of whether the binary classifier is logistic regression or a decision tree. Variance, however, follows mean root-to-leaf depth when logistic regression is used at each node, so flatter, more balanced NDs are more stable, whereas with decision-tree node models, variance is not associated with depth and differs little across ND structures. On six datasets with 6 to 19 classes, held-out accuracy, log loss, and prediction variance remain spatially organised on the same topology-only map. None of the six standard construction heuristics tested selects a tree from the best-performing cluster on any dataset, although some select high-performing individual trees. The ndscape embedding turns ND selection from one opaque construction step into a readable comparison of class groupings, exposing the upper-level splits that distinguish high-performing and stable regions.\\

\noindent Keywords: Multiclass classification, Nested dichotomies, Hierarchical reduction, Structural similarity, Embedding
\end{abstract}

\section{Introduction}

Multiclass classification is often reduced to a sequence of binary problems. Nested dichotomies (NDs) do this with a tree of binary classifiers \citep{FrankEibe2004Eond}. Compared to alternative reductions (such as one-vs-one and one-vs-all), NDs produce class probabilities by construction when the binary models are probabilistic \citep{FrankEibe2004Eond}, give an interpretable hierarchy, and can support counterfactual reasoning about predictions \citep{FdezSanchez2021}.

For instance, Figure~\ref{fig:single_nd} shows a possible ND structure for a 4-class problem. Each internal node partitions the classes into two disjoint groups and fits a binary classifier to separate them. The probability of a class is the product of the conditional probabilities along the root-to-leaf path. Considering the example in Figure~\ref{fig:single_nd}, $Y \in \{1, 2, 3, 4\}$ is the target variable. Also, consider $X$ as the design matrix. The binary splits are $\nd{1,2}{3,4}$, $\nd{1}{2}$, and $\nd{3}{4}$, that are represented by the binary variables $Z_{\nd{1,2}{3,4}} \in \{(1,2), (3,4)\}$, $Z_{\nd{1}{2}} \in \{1,2\}$, and $Z_{\nd{3}{4}} \in \{3,4\}$, respectively. Then, using the conditional probabilities, 
\begin{equation*}
\begin{aligned}
\Pr(Y = 1 \mid X)
&= \Pr\!\left(Z_{\nd{1,2}{3,4}} = (1,2) \mid X\right)\\
&\quad\times \Pr\!\left(Z_{\nd{1}{2}} = 1 \mid Z_{\nd{1,2}{3,4}} = (1,2),\, X\right).
\end{aligned}
\end{equation*}

\begin{figure}
    \centering
    \includegraphics[width=0.7\textwidth]{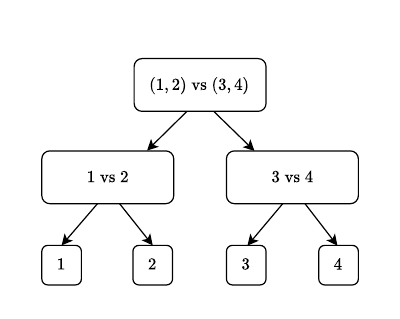}
    \caption{Example of a nested dichotomy structure for a dataset with four classes \{1, 2, 3, 4\}}
    \label{fig:single_nd}
\end{figure}

For four classes, the full ND space contains 15 structures. The number of possible ND structures grows quickly with the number of classes. For $M$ classes, there are $(2M-3)!!$ rooted binary trees, giving 945 structures for 6 classes, 10{,}395 for 7 classes and so on \citep{FrankEibe2004Eond}. Moreover, each structure can produce different predictions for the same data. That is, each structure is a different model, and finding the best structure (or best model) is a key challenge in ND modelling.

Existing work has addressed structural variability in two ways. One approach averages over many structures through ensembling \citep{FrankEibe2004Eond, rodriguez2010forests, Wever2018, evolutionary_ensembling2023}, and another searches for a single high-performing structure \citep{Dong2005, duarte2012nested, Leathart2016, Yang2017, Melnikov2018}. Neither approach provides insights into how structures relate to one another, nor why some structures perform similarly and others very differently. A practitioner who chooses one ND has no account of where their choice sits among the alternatives, or which class groupings are driving performance. To our knowledge, no prior work has analysed how the ND structure shapes bias, variance, and probability loss across the structure space.

We address this gap with ndscape, which organises the structures themselves. Before any model is fitted, ndscape represents each ND as a vector of depth-weighted structural tokens, measures the similarity between two structures by the cosine of the angle between their token vectors, and projects the structures into two dimensions using metric multidimensional scaling, so that structurally similar NDs lie close together. Because the map uses only topology, it is fixed for a given class count and any predictive quantity can be displayed on it.

The type of data determines what we display. With a known generating ND, the true class probabilities are available, so we compute squared bias, variance, and squared probability loss for each structure. On real data, these are unavailable, so we overlay held-out accuracy, log loss, and bootstrap prediction variance instead. In both settings, we ask whether structures that lie close on the map perform alike.

Our contributions are:
\begin{itemize}
    \item We show that predictive behaviour across the ND structure space is spatially organised. In simulations with known generating structures for 6 to 10 classes, squared bias, variance, and squared probability loss are each positively autocorrelated on a topology-only embedding.
    \item We introduce ndscape, a depth-weighted, token-based representation of ND structures that builds a two-dimensional embedding of the ND space, released as an open-source package \citep{ndscape}.
    \item We show that bias and variance are organised by different structural features. Squared bias and squared probability loss are lowest among the structures most similar to the data-generating ND.  With logistic regression as the binary model, variance tracks the mean root-to-leaf path length, with balanced NDs being more stable. Replacing the logistic node model with a decision tree keeps the bias and loss organisation but flattens variance, which then tracks bias rather than depth (Appendix~\ref{appendix:dt}).
    \item We present a study on six UCI datasets with 6 to 19 classes that overlays held-out accuracy, log loss, and prediction variance on the embedding. The tested construction heuristics select strong individual NDs, but none reaches the best-performing regions of the structure space.
\end{itemize}

\section{Literature review}\label{sec:lit_rev}

\cite{FrankEibe2004Eond} introduced nested dichotomies. Because the space of structures is large, several heuristics have been constructed to select individual NDs. Random-pair selection (RPND) \citep{Leathart2016} chooses two classes at random at each node, fits a temporary binary classifier to separate them, and assigns each remaining class to the side predicted for most of its instances. Best-of-K (BoK) \citep{Yang2017, Melnikov2018} samples $K$ random ND structures and returns the one with the lowest estimated loss. Biased best-of-K (BBoK) extends BoK with a sampler biased toward balanced interior nodes, meaning the same number of classes on each side of the split \citep{Melnikov2018}. Agglomerative clustering ND (ACND) \citep{Melnikov2018} merges class clusters in a bottom-up manner based on similarity. Class-balanced ND (CBND) \citep{Dong2005} builds the ND top-down, splitting the classes at each node into two groups as close to equal in size as possible. Clustering-based ND (NDC) \citep{duarte2012nested} instead clusters the classes at each node from their feature-space representation and splits on that clustering. These methods differ in how they construct the ND, but all focus on finding the best ND. This leaves open the question of how differences among possible ND structures relate to prediction performance.

Some works discuss averaging over many ND structures rather than selecting one. \cite{FrankEibe2004Eond} proposed an ensemble of randomly sampled NDs as a hedge against a poor structural choice. This was extended to forests of NDs \citep{rodriguez2010forests} and to evolutionary ensembles \citep{Wever2018, evolutionary_ensembling2023}. However, neither the selection heuristics nor the ensembles examine why two structures trained on the same data can perform so differently.

A standard way to examine the difference between the ND structures is the bias-variance decomposition of loss. Bias is the gap between the expected fitted probability and the truth, while variance measures how much the fitted probability changes across training samples \citep{geman1992,kohavi1995study,hastie2009elements}. We use the standard multiclass form of this decomposition \citep{geman1992, domingos2000unified}. This is the first work to decompose the loss of an ND into bias and variance, and to map how the two vary across the space of structures.

Mapping how performance varies across the structure space calls for a way to represent and compare NDs, and the relevant tools come from outside the ND literature. The vector space model of information retrieval represents each document as a vector of weighted term counts and measures the similarity of two documents by the cosine of the angle between their vectors \citep{Salton1975, Salton1988}. Metric multidimensional scaling takes a matrix of pairwise dissimilarities and places the objects in a low-dimensional space whose Euclidean distances reproduce those dissimilarities as closely as possible \citep{Torgerson1952}. Moran's $I$, from spatial statistics, measures spatial autocorrelation, the tendency of values at nearby locations to resemble one another more than distant ones \citep{Moran1950}. We use these tools to compare NDs directly.

\section{Method}

We first present the method for simulated data, where the generating model is known, and every quantity of interest can be computed, allowing the bias-variance behaviour of the structures to be studied directly. We then describe how the method adapts when the generating model is unknown.

The method has three steps:

\begin{description}
\item[First,] we embed each ND structure for a fixed class count using its topology, without any data or fitted models (Section~\ref{sec:embedding}).
\item[Second,] on simulated data with a known generating ND, we train each candidate structure many times and estimate its squared bias, variance, and squared probability loss (Section~\ref{sec:biasvariance}).
\item[Third,] we attach these values to the embedding coordinates and measure whether nearby NDs behave alike using Moran's $I$ and related spatial summaries (Section~\ref{sec:spatial_summaries}).
\end{description}

For real datasets where the true probabilities are unknown, we replace the simulated quantities with empirical held-out estimates and repeat the spatial analysis (Section~\ref{sec:real_data}). The ndscape embedding is built once per class count and is reused across all of these steps.

\subsection{Embedding ND structures}\label{sec:embedding}

We present an ND embedding based on vector-space methods from information retrieval \citep{Salton1975, Salton1988}. The embedding is built in three stages. First, each ND is encoded as a collection of structural tokens. Second, these tokens are converted into normalised count vectors and compared using cosine similarity. Third, the resulting pairwise distances are embedded in two-dimensional Euclidean space. Figure~\ref{fig:nd-embedding} illustrates these stages using the ND from Figure~\ref{fig:single_nd}.

\begin{figure}
    \centering
    \includegraphics[width=\linewidth]{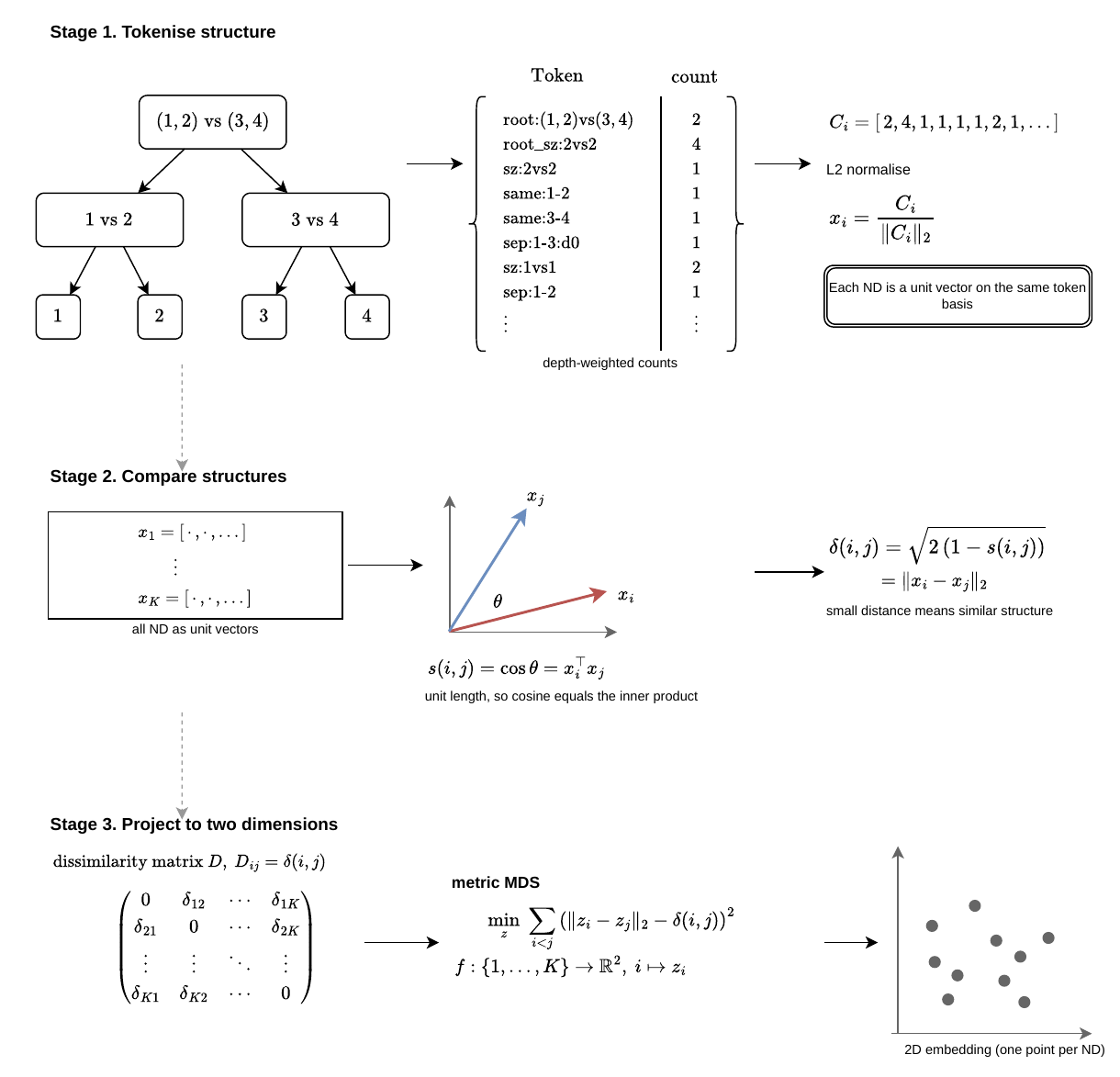}
    \caption{An example of the three steps to embed the ND space into 2D}
    \label{fig:nd-embedding}
\end{figure}

For stage 1, we represent each ND using tokens derived from its internal nodes. For an $M$-class problem, let $d \ge 0$ denote node depth from the root, and let $d_{\max} = M-2$ be the maximum depth, as that is the largest possible internal-node depth for an $M$-class binary tree. 

For each internal split at depth $d \le d_{\max}$ with child sets $L$ and $R$, we sort the classes within each child and order the two children by size, larger first. This gives a canonical ordered pair $(A,B)$, so mirrored nodes yield identical tokens. From each ordered split we generate the token families below, each written as a literal string. In a token, \texttt{|A|} and \texttt{|B|} are the numbers of classes on the two sides $A$ and $B$, \texttt{i} and \texttt{j} are single classes from opposite sides, and a suffix \texttt{d}$k$ records the node depth $k$ capped at 5, so the labels run \texttt{d0} to \texttt{d5} and any deeper node also uses \texttt{d5}. A token's multiplicity is the number of times it is added to the count vector, and the depth weight $w_d = \max(1, \lfloor (d_{\max} + 1 - d)/2 \rfloor)$ is largest at the root and decreases with depth, so upper-level splits count for more.

\begin{description}
  \item[Size pattern.] The token \texttt{sz:|A|vs|B|}, the size of each side, with multiplicity $w_d$. A root split of classes 1, 2 against 3, 4 gives \texttt{sz:2vs2}.
  \item[Separation pairs.] For each class $i$ on side $A$ and class $j$ on side $B$, two tokens that both record $i$ and $j$ being separated. The first, \texttt{sep:i-j:sz:|A|vs|B|}, also notes the split size and has multiplicity $w_d$. The second, \texttt{sep:i-j:d}$k$, notes the depth and has multiplicity 1. For that split the pair $1,3$ gives \texttt{sep:1-3:sz:2vs2} and \texttt{sep:1-3:d0}.
  \item[Co-membership.] For each pair of classes $i$ and $j$ kept on the same side, the token \texttt{same:i-j} with multiplicity 1, recording which classes the split groups together rather than separates. Classes 1 and 2 give \texttt{same:1-2}.
  \item[Root emphasis.] At the root ($d = 0$) only, three further families with extra weight. The full split \texttt{root:AvsB}, listing the classes on each side, with multiplicity 2; the root size \texttt{root\_sz:|A|vs|B|}, the counts, with multiplicity 4; and root copies of the separation tokens with multiplicity 2. The same split gives \texttt{root:1,2vs3,4} and \texttt{root\_sz:2vs2}.
  \item[Small-group emphasis ($M \ge 7$).] For ~ every ~ side ~ $S$ ~ holding ~ 2 ~ or ~ 3 ~ classes, ~ the ~ token  \texttt{small\_group\_size:|S|:d}$k$ with multiplicity $M - 6$, so small isolated clusters keep a stable weight as the number of classes grows. A two-class side split off at depth 2 gives \texttt{small\_group\_size:2:d2}.
\end{description}

The multiplicities encode two structural assumptions. First, upper-level splits affect more downstream probabilities and therefore receive greater weight. Second, the root partition receives separate emphasis because it is the only split shared by all class probabilities. Pair tokens record the class separations and co-memberships that determine the binary subproblems induced by the tree.

Stage 2 converts these tokens into a common vector space. Let $C_{i,t}$ denote the count of token $t$ in ND $i$. For each ND $i$, we form the vector $C_i$ and normalise to unit length
\begin{equation}
x_i = \frac{C_i}{\lVert C_i \rVert_2}.
\end{equation}

This tokenisation is implemented in \texttt{scikit-learn} with \texttt{use\_idf=False} and L2 normalisation \citep{scikit-learn}. 

We compare NDs using cosine similarity, which, after L2 normalisation, reduces to the inner product
\begin{equation}
s(i,j) = \cos\angle(x_i, x_j) = x_i^\top x_j,
\end{equation}
and we convert similarity to dissimilarity using the chord distance,
\begin{equation}
\delta(i,j) = \sqrt{2(1 - s(i,j))}.
\end{equation}
For unit vectors $x_i$ and $x_j$, $\delta(i,j) = \lVert x_i - x_j \rVert_2$, so the dissimilarity is exactly the Euclidean distance between the normalised vectors. This aligns the dissimilarity with the Euclidean objective, which is minimised at the next stage \citep{kruskal1964}.

For stage 3, metric multidimensional scaling places each ND at a point $z_i \in \mathbb{R}^2$ so that the Euclidean distance between any two points matches their dissimilarity $\delta(i,j)$ as closely as possible. The points minimise
\begin{equation}
\sum_{i<j}\left(\lVert z_i - z_j\rVert_2 - \delta(i,j)\right)^2 .
\end{equation}

We use metric MDS rather than a neighbourhood-based projection because the spatial statistics in Section~\ref{sec:spatial_summaries} compare each structure with its neighbours and require that absolute distances be preserved. Methods such as t-SNE and UMAP preserve local proximity but distort longer-range distances, which would bias Moran's $I$. We assess the projection's trustworthiness at $k=20$ by the fraction of each point's 20 nearest neighbours in token space that remain among its 20 nearest neighbours in two dimensions. The map is fixed for a class count and uses no other data, fitted models, or performance values.

The \texttt{Python} implementation is available at \url{https://github.com/res-lucid/ndscape} \citep{ndscape}.
\subsection{Bias and variance analysis}\label{sec:biasvariance}

A multiclass model can fail in two ways. Its predictions can be systematically wrong on average, or they can swing widely from one training sample to the next. Loss alone cannot tell these apart. As a standard procedure, we therefore decompose loss into squared bias and variance and analyse them separately \citep{hastie2009elements,friedman1997bias}.

For a multiclass problem, let
\[
p_m(x)=P(Y=m\mid X=x)
\]
be the true probability of class \(m\) at input point \(x\), and let \(\hat p_m(x)\) be the fitted probability for the same class. We define the classwise squared probability loss by
\begin{equation}
L_m(x)=\left(\hat p_m(x)-p_m(x)\right)^2.
\end{equation}
The corresponding squared bias is
\begin{equation}
\operatorname{Bias}_m^2(x)=\left(\mathbb E[\hat p_m(x)]-p_m(x)\right)^2,
\end{equation}
with variance
\begin{equation}
\operatorname{Var}_m(x)=
\mathbb E\!\left[\left(\hat p_m(x)-\mathbb E[\hat p_m(x)]\right)^2\right].
\end{equation}
For class \(m\), the expected squared probability loss is
\begin{equation}
\mathbb E[L_m(x)] = \operatorname{Bias}_m^2(x)+\operatorname{Var}_m(x).
\end{equation}

A proof is given in Appendix~\ref{appendix:bias_variance}. Summing over classes \(\{1,\dots,M\}\) gives the multiclass squared probability loss
\begin{equation}\label{eq:loss}
L(x)=\sum_{m=1}^M L_m(x),
\end{equation}
with decomposition
\begin{equation}
\mathbb E[L(x)]
=
\sum_{m=1}^M \operatorname{Bias}_m^2(x)
+
\sum_{m=1}^M \operatorname{Var}_m(x).
\end{equation}
This decomposition applies to any multiclass probabilistic model, including an ND.

In our simulation experiments, the true class probabilities are known from the data-generating process, and expectation can be approximated by repeatedly generating training datasets \citep{morris2019using,domingos2000unified}. For a fixed ND structure \(T\), let \(D_1,\dots,D_R\) be \(R\) independently generated training datasets, and let \[ \hat p^{(T)}_{D_r,m}(x) \] denote the fitted probability for class \(m\) at evaluation point \(x\) from the model trained on dataset \(D_r\). Replacing expectations by averages over \(D_1,\dots,D_R\), we approximate the classwise squared bias and variance at \(x\) by
\begin{equation}\label{eq:bias}
\widehat{\operatorname{Bias}}_m^2(T,x)
=
\left(
\frac{1}{R}\sum_{r=1}^R \hat p^{(T)}_{D_r,m}(x)-p_m(x)
\right)^2,
\end{equation}
and
\begin{equation}
\widehat{\operatorname{Var}}_m(T,x)
=
\frac{1}{R}\sum_{r=1}^R
\left(
\hat p^{(T)}_{D_r,m}(x)
-
\frac{1}{R}\sum_{s=1}^R \hat p^{(T)}_{D_s,m}(x)
\right)^2.
\end{equation}\label{eq:var}
The corresponding multiclass loss is estimated at \(x\) by
\begin{equation}
\widehat{L}(T,x)
=
\sum_{m=1}^M \widehat{\operatorname{Bias}}_m^2(T,x)
+
\sum_{m=1}^M \widehat{\operatorname{Var}}_m(T,x).
\end{equation}

Following general principles for simulation design \citep{morris2019using}, we hold the data-generating process fixed within each simulation and vary only the sampled training data. The generating ND, the feature-generation mechanism, and the sample size remain unchanged across repetitions, and the evaluation set is generated once and then held fixed. This isolates sampling variation and allows probability-based estimates of bias and variance to be calculated. In the experiments reported here, we independently generate training datasets of size \(n=5{,}000\), together with a common evaluation set of size \(N=20{,}000\). 

To choose the number of repetitions ($R$), we examined the relative standard error of the estimated performance metrics. Using the loss in Equation~\eqref{eq:loss}, define

\begin{equation}
\label{eq:Ar}
\mathcal{L}_r(T)=\frac{1}{N}\sum_{i=1}^N\sum_{m=1}^M\left(\hat p^{(T)}_{D_r,m}(x_i)-p_m(x_i)\right)^2,
\end{equation}
where $r=1,\ldots,R$ indexes the training-data repetition. Equation \eqref{eq:Ar} is the evaluation-set average squared probability loss for structure \(T\) in repetition \(r\), and its average is,
\[
\bar{\mathcal{L}}_R(T)=\frac{1}{R}\sum_{r=1}^R \mathcal{L}_r(T).
\]

The Monte Carlo standard error of $\bar{\mathcal{L}}_R(T)$ is estimated by,
\[
\widehat{\mathrm{SE}}(\bar{\mathcal{L}}_R(T))=\frac{s_\mathcal{L}}{\sqrt{R}},
\]
where \(s_\mathcal{L}\) is the standard deviation of \(\mathcal{L}_1(T),\dots,\mathcal{L}_R(T)\). We then report the relative standard error 
\[
\mathrm{RelSE}_{\mathcal{L}}(R)
=
\frac{\widehat{\mathrm{SE}}(\bar{\mathcal{L}}_R(T))}{\bar{\mathcal{L}}_R(T)}.
\]

For squared bias and variance, we estimate Monte Carlo error by bootstrapping the completed repetitions \citep{Efron1981}. Let \(\widehat{\theta}_R(T)\) denote either the estimated squared bias or the estimated variance for structure \(T\) computed from \(R\) repetitions, and let \(\widehat{\mathrm{SE}}_{\mathrm{boot}}(\widehat{\theta}_R(T))\) be the corresponding bootstrap standard error. We then report,

\[
\mathrm{RelSE}_{\mathrm{Bias}^2}(R)
=
\frac{\widehat{\mathrm{SE}}_{\mathrm{boot}}(\widehat{\mathrm{Bias}^2}_R(T))}{\widehat{\mathrm{Bias}^2}_R(T)}, \textnormal{ and }
\mathrm{RelSE}_{\mathrm{Var}}(R)
=
\frac{\widehat{\mathrm{SE}}_{\mathrm{boot}}(\widehat{\mathrm{Var}}_R(T))}{\widehat{\mathrm{Var}}_R(T)}.
\]

For precise estimates, we use a 1\% relative standard error threshold as a conservative criterion for choosing \(R\).

For each ND structure \(T\), the final bias, variance, and loss values are obtained by averaging these quantities over the evaluation set \(x_1,\dots,x_N\). We attach these ND-level values to the embedding coordinates of \(T\). The embedding is a structural map of the ND space, and colouring it by squared bias, variance, or squared probability loss shows how each varies across that space.

\subsection{Spatial summaries}\label{sec:spatial_summaries}

The ndscape embedding maps the ND space to two dimensions, and two NDs are near one another when they lie close on the map. To measure whether nearby NDs have similar performance, we compute global Moran's $I$,
\[
I = \frac{n}{S_0} \frac{\sum_i \sum_j w_{ij}(x_i - \bar{x})(x_j - \bar{x})}{\sum_i (x_i - \bar{x})^2}, \qquad S_0 = \sum_i \sum_j w_{ij},
\]

where $n$ is the number of NDs, $x_i$ is the metric value for ND $i$, $\bar{x}$ is its mean, and $w_{ij}$ is the row-standardised neighbourhood weight between NDs $i$ and $j$ \citep{Moran1950}. Significance is assessed by 999 random permutations of the metric values over the fixed embedding locations. We report the upper-tail pseudo $p$-value,
\[
p = \frac{m+1}{1000},
\]
where $m$ is the number of permutations with Moran's $I$ as large or larger than the observed value. The smallest reportable value is therefore $p=0.001$.

A positive Moran's I means nearby NDs have similar values, near zero means no spatial pattern, and negative means neighbours differ. We compute Moran's $I$ separately for squared bias, variance, and squared probability loss. Neighbourhoods are defined using a directed $k$-nearest-neighbour graph with $k = 50$, followed by row standardisation so each row sums to one. We use the same $k$ at every class count so the statistic is comparable across runs. For the enumerated 6 and 7-class spaces, $k = 50$ is a neighbourhood of roughly $5\%$ and $0.5\%$ of the space, large enough for a stable local average yet small enough to remain local, and a sweep that increases $k$ from $10$ in steps of $5$ shows both graphs first become connected near $k = 50$. For the sampled 8 to 10-class spaces of 15{,}000 NDs, the same $k$ is an even smaller fraction.

We also report the bivariate version of Moran's $I$, which we use as a spatial-lag association between one variable at $i$ and the spatial lag of another variable at the neighbours of $i$,
\[
I_{xy} = \frac{n}{S_0} \frac{\sum_i \sum_j w_{ij}(x_i - \bar{x})(y_j - \bar{y})}{\sqrt{\sum_i (x_i - \bar{x})^2 \sum_i (y_i - \bar{y})^2}}.
\]
For our main use of this statistic, $x_i$ is the variance at ND $i$ and $y_i$ is the mean root-to-leaf depth at ND $i$, so $I_{xy}$ measures whether an ND with high variance tends to sit near NDs with high depth. We do not interpret it as a pointwise correlation between variance and depth for the same tree.

A second summary measures whether the lowest-decile structures for each metric lie near the generating ND, rather than only near each other. For each simulation, we identify the \(10\%\) of candidate structures nearest to the generating ND in the embedding and report the fraction of the relevant lowest decile that falls inside this set. Under no structural relationship, this fraction is \(10\%\). We also report the AUC of distance to the generating ND as a ranking signal for the lowest decile, where \(0.5\) indicates no relationship.

\section{Simulation study}\label{sec:simulation_study}
This section describes the simulation design (Section~\ref{sec:simulation_design}), an ablation comparing structural representations (Section~\ref{sec:ablation}), and the main results (Section~\ref{sec:simulation_results}). We evaluate all structures at 6 and 7 classes, sample structures at 8 to 10 classes, and report the two regimes together.

\subsection{Simulation design}\label{sec:simulation_design}
 
The simulation study tests whether structural proximity in the ndscape embedding corresponds to similar bias, variance, and squared probability loss. For 6 and 7-class problems, every possible structure can be evaluated. At higher class counts, full enumeration becomes impractical. The 8-class space contains $13!! = 135{,}135$ structures, rising to $15!! = 2{,}027{,}025$ at 9 classes and $17!! = 34{,}459{,}425$ at 10 classes. In these settings, we sample 15{,}000 structures per generator. Appendix~\ref{appendix:sampling} gives the complete sampling procedure. 

The main simulation setting has 1000 features, with only 50 features used at each node. For 6 and 7 classes, we also include two dense variants, one with 10 features and one with 1000 features, in which every feature is used at every node. Appendix~\ref{appendix:feature_gen} reports these variants and gives the full feature distributions. For each class count, we construct the ndscape embedding, cluster the embedded structures by $k$-means with the number of clusters chosen by the silhouette coefficient \citep{rousseeuw1987}, and select generating NDs stratified across clusters. For 6 and 7 classes, we use 50 generating structures, and for 8, 9, and 10 classes, we use 5. Each generating ND fixes the true class probabilities through a rule attached to every internal node. For the main simulations, each internal node of the generating ND uses a logistic model to assign a probability to one of its two child branches, and class probabilities are obtained by multiplying the branch probabilities along each root-to-leaf path. Appendix~\ref{appendix:target_gen} gives the decision-tree generator used for the sensitivity analysis, which follows the same recursion. Candidate NDs are fitted to data sampled from this fixed generator. Table~\ref{tab:simulation_counts} shows that the main study fits 79.2 million NDs, which is the reason for using fewer generating structures for the 8- to 10-class runs.

\begin{table}[t]
\centering
\caption{Number of fitted ND structures in the simulation study. Each ND is fitted with a logistic regression model at each internal node}
\label{tab:simulation_counts}
\small
\renewcommand{\arraystretch}{1.28}
\begin{tabular*}{\textwidth}{@{\extracolsep{\fill}}ccccc@{}}
\toprule
Classes & ND per generator & Generators & Repetitions & ND model fits\\
\midrule
6  & 945    & 50 & 100 & 4{,}725{,}000\\
7  & 10{,}395 & 50 & 100 & 51{,}975{,}000\\
8  & 15{,}000 & 5  & 100 & 7{,}500{,}000\\
9  & 15{,}000 & 5  & 100 & 7{,}500{,}000\\
10 & 15{,}000 & 5  & 100 & 7{,}500{,}000\\
\midrule
Total & -- & -- & -- & 79{,}200{,}000\\
\bottomrule
\end{tabular*}
\end{table}

For each generator, we draw independent training datasets of size $n = 5{,}000$ and hold a common evaluation set of size $N_{\text{eval}} = 20{,}000$ fixed. These sizes were chosen to stabilise the binary node fits and reduce noise in the evaluation-set averages. The binary models in the main study are \texttt{scikit-learn} binary logistic regression with L2 regularisation, the \texttt{newton-cholesky} solver, $C = 0.1$, and \texttt{max\_iter = 2000} \citep{scikit-learn}. We also present similar results using the decision tree base learner in Appendix~\ref{appendix:dt}.

Pilot runs were used to choose the number of repetitions $R$. In the 6-class pilot, the $1\%$ relative Monte Carlo error threshold was reached by $R = 10$ for squared probability loss, $R = 30$ for squared bias, and $R = 40$ for variance, and the 7-class pilot showed the same pattern. The pilot curves are shown in Appendix~\ref{appendix:mc_error}. We set $R = 100$ for the main study, which keeps the Monte Carlo uncertainty below $1\%$ for all three measures and leaves headroom so that the 8- to 10-class estimates, where fewer generators are averaged, remain comparably precise. 

\subsection{Ablation study}\label{sec:ablation}

The ablation compares three structural representations on the 1000-feature sparse simulation in the 6 and 7-class setting. The depth-weighted tokeniser is the representation from Section~\ref{sec:embedding}. Bag-of-words (BoW) uses the same token vocabulary without depth-weighted multiplicities, so every token contributes once regardless of depth. The Robinson--Foulds distance (RF) counts the bipartitions present in one ND but absent from the other \citep{Robinson1981}. We include RF because it is the standard topological distance between phylogenetic trees, which makes it the natural depth-agnostic baseline to test against. For RF, we convert each rooted ND into its set of internal bipartitions and use the symmetric difference of these sets as the dissimilarity measure. 

We also measure whether each representation places low-error structures near the generating ND. For each generator, candidate structures are ranked by their structural distance to the generator. AP@10 is the average precision over the ten nearest NDs, with structures in the lowest-error decile treated as relevant. Each representation is projected into two dimensions using metric MDS and assessed using Moran's $I$ and AP@10. Because all three representations are evaluated on the same 50 generators, the comparison is paired, and Table~\ref{tab:ablation-1000-sparse} reports the mean and standard deviation across generators.

The depth-weighted tokeniser has the highest mean in every row at both class counts. The Robinson--Foulds embedding is the weakest representation across all rows, consistent with its insensitivity to depth and class grouping. The depth-weighted tokeniser also outperforms BoW. This supports the use of depth-weighted token multiplicities.

\begin{table}[!htbp]
\centering
\caption{Ablation across three structural representations (tokeniser, BoW, Robinson--Foulds) in the 1000-feature sparse simulation. Values are mean $\pm$ SD over 50 generators. SPL denotes squared probability loss. Bold marks the row maximum. See Section~\ref{sec:spatial_summaries} for Moran's $I$ and Section~\ref{sec:ablation} for AP@10 definitions}
\label{tab:ablation-1000-sparse}
\small
\renewcommand{\arraystretch}{1.28}
\begin{tabular*}{\textwidth}{@{\extracolsep{\fill}}lccc@{}}
\toprule
Measure & Tokeniser & BoW & RF \\
\midrule
\multicolumn{4}{l}{\textit{6 classes}} \\
Moran $I$ (bias$^2$) & \textbf{0.356 $\pm$ 0.155} & 0.256 $\pm$ 0.148 & 0.171 $\pm$ 0.163 \\
Moran $I$ (var) & \textbf{0.541 $\pm$ 0.044} & 0.271 $\pm$ 0.048 & 0.033 $\pm$ 0.019 \\
Moran $I$ (SPL) & \textbf{0.300 $\pm$ 0.138} & 0.248 $\pm$ 0.148 & 0.178 $\pm$ 0.168 \\
AP@10 (bias$^2$) & \textbf{0.372 $\pm$ 0.138} & 0.275 $\pm$ 0.119 & 0.265 $\pm$ 0.115 \\
AP@10 (SPL) & \textbf{0.379 $\pm$ 0.124} & 0.294 $\pm$ 0.119 & 0.273 $\pm$ 0.113 \\
\midrule
\multicolumn{4}{l}{\textit{7 classes}} \\
Moran $I$ (bias$^2$) & \textbf{0.473 $\pm$ 0.119} & 0.256 $\pm$ 0.107 & 0.188 $\pm$ 0.120 \\
Moran $I$ (var) & \textbf{0.726 $\pm$ 0.048} & 0.214 $\pm$ 0.034 & 0.045 $\pm$ 0.018 \\
Moran $I$ (SPL) & \textbf{0.378 $\pm$ 0.103} & 0.263 $\pm$ 0.121 & 0.208 $\pm$ 0.121 \\
AP@10 (bias$^2$) & \textbf{0.297 $\pm$ 0.140} & 0.237 $\pm$ 0.078 & 0.179 $\pm$ 0.059 \\
AP@10 (SPL) & \textbf{0.297 $\pm$ 0.119} & 0.245 $\pm$ 0.081 & 0.191 $\pm$ 0.062 \\
\bottomrule
\end{tabular*}
\end{table}

\subsection{Simulation results}\label{sec:simulation_results}

Squared bias and squared probability loss are spatially organised at every class count (see Table~\ref{tab:scaling}). Squared bias has a mean Moran's $I$ of $0.356$ at 6 classes, rising to $0.473$ at 7 classes and staying between $0.34$ and $0.39$ from 8 to 10 classes. Squared probability loss follows the same shape at a lower level, from $0.300$ at 6 classes to $0.378$ at 7 and between $0.26$ and $0.31$ from 8 to 10 classes. All Moran's $I$ values in Table~\ref{tab:scaling} are positive and statistically significant under permutation testing, with $p \le 0.001$.

\begin{table}[t]
\centering
\caption{Spatial autocorrelation in the ndscape embedding (depth-weighted tokeniser) for 6 to 10 classes, 1000-feature sparse setting. $I$ is univariate Moran's $I$, and $I_{xy}$ is the bivariate Moran's $I$ between variance and mean root-to-leaf depth. The 6- and 7-class rows are averaged over 50 generators and the 8, 9, and 10-class rows over 5. Cells report mean $\pm$ SD, and every $I$ is significant at $p \le 0.001$} 
\label{tab:scaling}
\small
\renewcommand{\arraystretch}{1.28}
\begin{tabular*}{\textwidth}{@{\extracolsep{\fill}}lcccc@{}}
\toprule
Classes & $I$ (bias$^2$) & $I$ (loss) & $I$ (var) & $I_{xy}$ (var, depth) \\
\midrule
6 & 0.356 $\pm$ 0.155 & 0.300 $\pm$ 0.138 & 0.541 $\pm$ 0.044 & 0.589 $\pm$ 0.077 \\
7 & 0.473 $\pm$ 0.119 & 0.378 $\pm$ 0.103 & 0.726 $\pm$ 0.048 & 0.762 $\pm$ 0.067 \\
8 & 0.393 $\pm$ 0.109 & 0.310 $\pm$ 0.086 & 0.791 $\pm$ 0.049 & 0.820 $\pm$ 0.056 \\
9 & 0.374 $\pm$ 0.098 & 0.297 $\pm$ 0.049 & 0.825 $\pm$ 0.055 & 0.865 $\pm$ 0.048 \\
10 & 0.342 $\pm$ 0.051 & 0.259 $\pm$ 0.022 & 0.833 $\pm$ 0.027 & 0.883 $\pm$ 0.025 \\
\bottomrule
\end{tabular*}
\end{table}

Table~\ref{tab:nearest-generator-enrichment} shows whether the lowest-error structures are concentrated near the generating ND. Under no relationship, the nearest 10\% of structures should contain 10\% of the lowest-error decile. At 6 classes, the nearest $10\%$ of structures hold $38.3\%$ of the lowest-bias decile and $39.1\%$ of the lowest-loss decile, both far above the $10\%$ expected under no relationship. The proportion decreases as the class count grows, reaching $8.7\%$ for bias and $9.9\%$ for loss at 10 classes, where the proximity signal is absent. The 8 to 10-class spaces are sampled to 15{,}000 NDs, so the local neighbourhood of the generator is represented more coarsely than in the enumerated spaces. The AUC of distance to the generator tells the same story, falling from about $0.72$ at 6 classes toward $0.5$ at 10.

\begin{table}[t]
\centering
\caption{Fraction of the lowest-decile structures captured near the generating ND, 1000-feature sparse setting. Recall is the fraction of each lowest decile inside the nearest 10\% of structures. AUC ranks each lowest decile by structural distance to the generator. Null values are 10\% and 0.5. Scores are the mean across 50 generators for 6- and 7-class problems, and across 5 generators for the others}
\label{tab:nearest-generator-enrichment}
\small
\renewcommand{\arraystretch}{1.28}
\begin{tabular*}{\textwidth}{@{\extracolsep{\fill}}lcccccc@{}}
\toprule
& \multicolumn{2}{c}{Squared bias} & \multicolumn{2}{c}{Variance} & \multicolumn{2}{c}{Squared prob.\ loss} \\
\cmidrule(lr){2-3} \cmidrule(lr){4-5} \cmidrule(lr){6-7}
Classes & Recall & AUC & Recall & AUC & Recall & AUC \\
\midrule
6 & 38.3\% & 0.716 & 26.1\% & 0.568 & 39.1\% & 0.711 \\
7 & 33.9\% & 0.678 & 19.1\% & 0.502 & 34.8\% & 0.671 \\
8 & 22.3\% & 0.604 & 16.0\% & 0.510 & 24.6\% & 0.611 \\
9 & 15.7\% & 0.559 & 6.9\% & 0.427 & 16.3\% & 0.582 \\
10 & 8.7\% & 0.478 & 13.3\% & 0.506 & 9.9\% & 0.494 \\
\bottomrule
\end{tabular*}
\end{table}

Figures~\ref{fig:6-class-bias-four-generators} and~\ref{fig:seven-class-bias-plots} give examples for the 6- and 7-class analyses, respectively. The point layout is fixed because the ndscape embedding depends only on topology. For both figures, when the generator moves, the low-bias region follows. The corresponding 8, 9, and 10-class maps are in Appendix~\ref{appendix:highclass} (Figures~\ref{fig:8-class-bias}, \ref{fig:9-class-bias}, and~\ref{fig:10-class-bias}) and show the same behaviour.

\begin{figure*}[t]
\centering
\setlength{\tabcolsep}{0pt}
\renewcommand{\arraystretch}{0}
\begin{minipage}[c]{0.91\textwidth}
\centering
\begin{tabular}{@{}cc@{}}
\includegraphics[height=0.225\textheight,trim={1.0cm 0.8cm 2.6cm 0.1cm},clip]{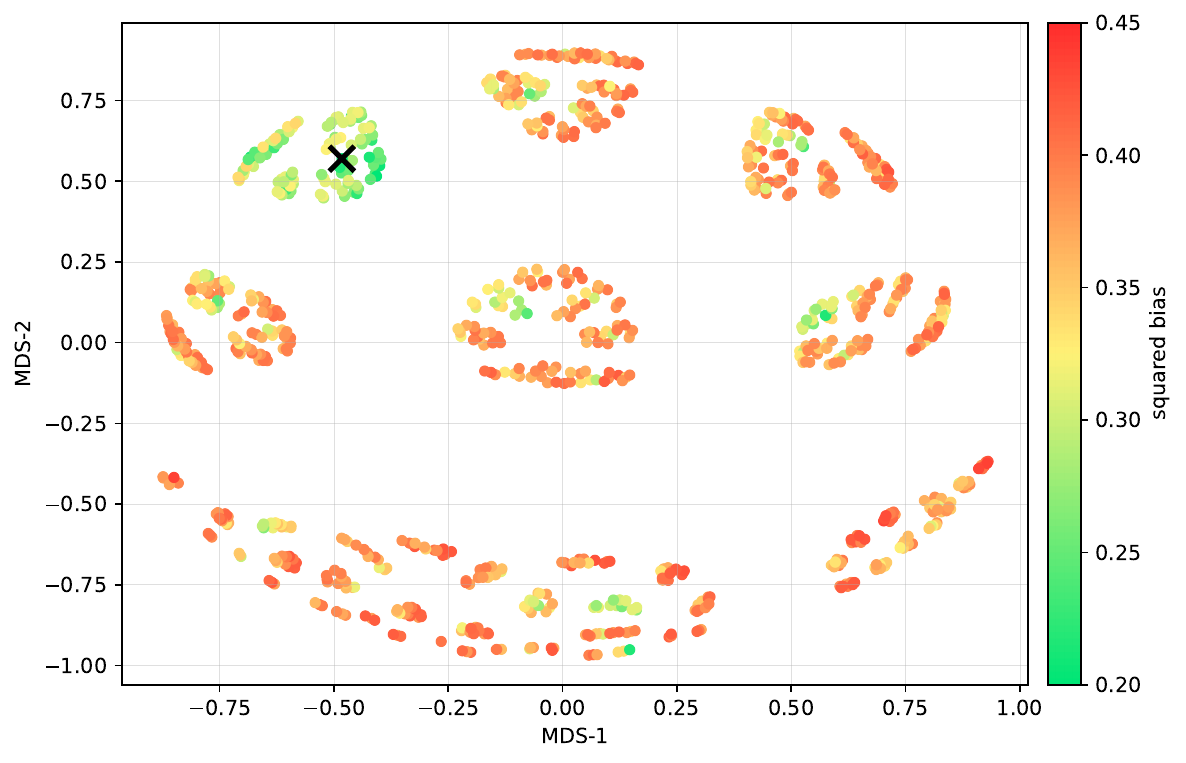}
&
\includegraphics[height=0.225\textheight,trim={1.8cm 0.8cm 2.6cm 0.1cm},clip]{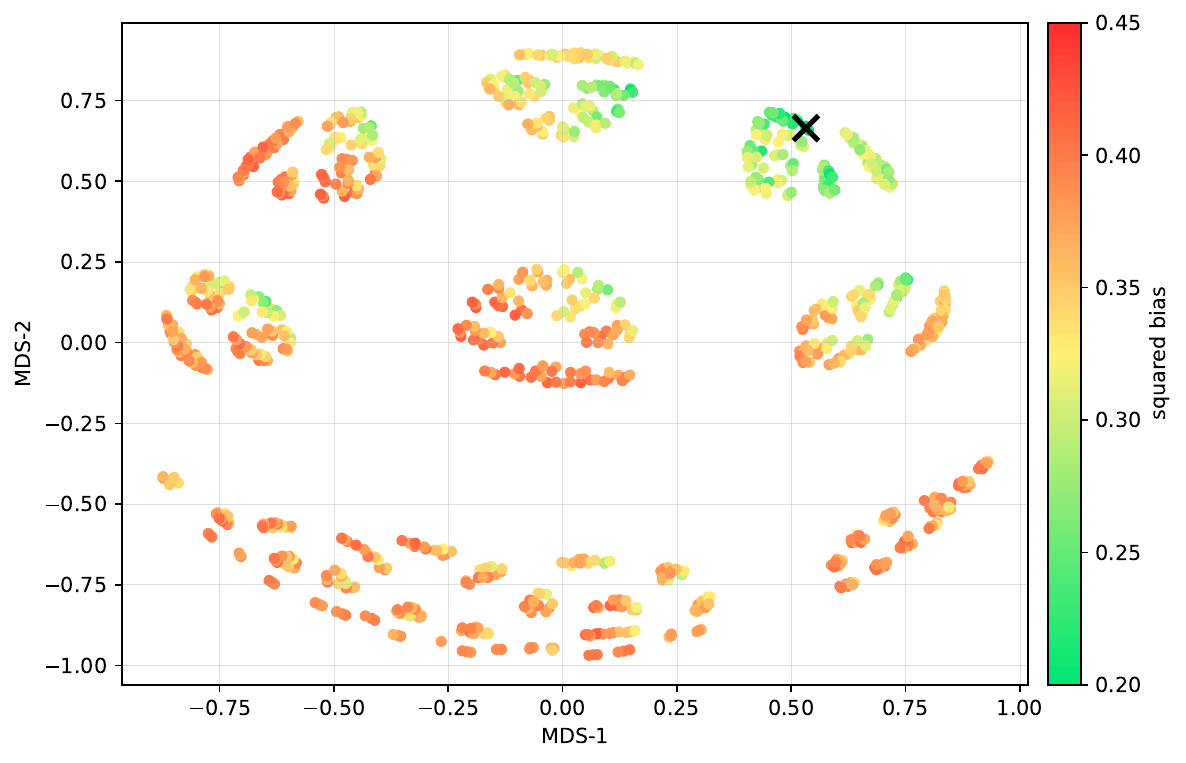}
\\[-0.18cm]
\includegraphics[height=0.225\textheight,trim={1.0cm 0.8cm 2.6cm 0.1cm},clip]{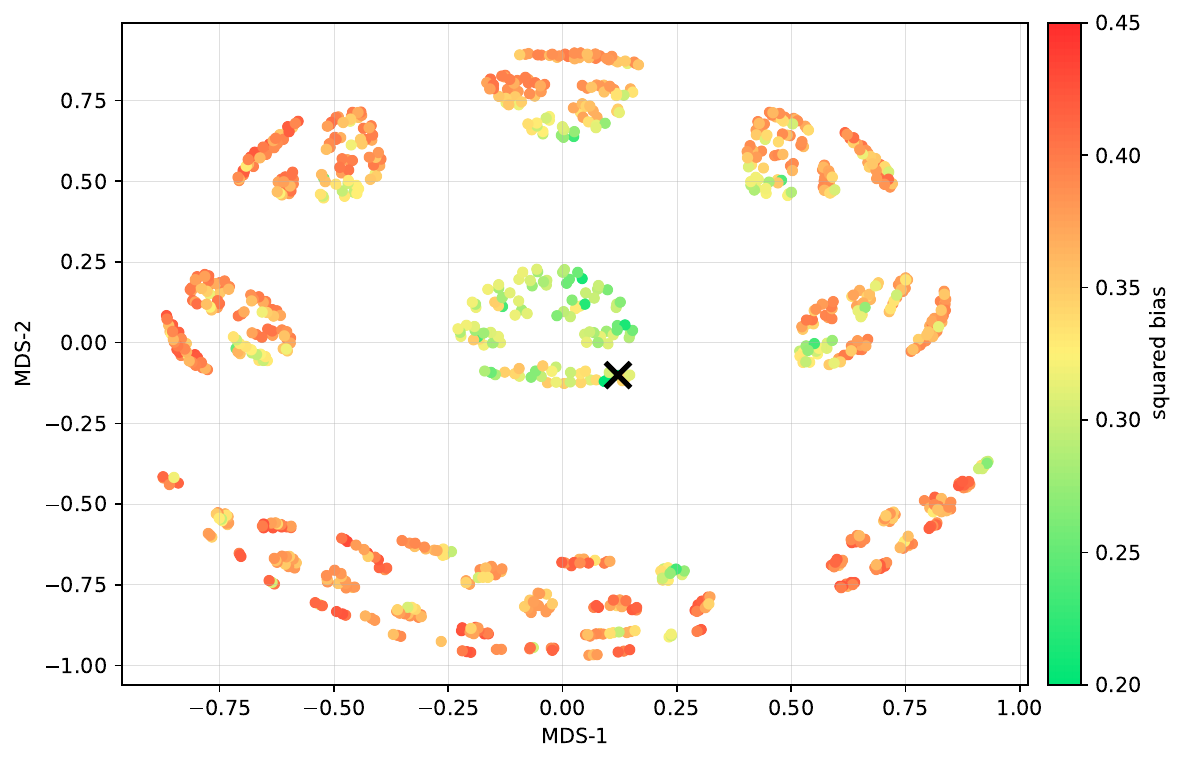}
&
\includegraphics[height=0.225\textheight,trim={1.8cm 0.8cm 2.6cm 0.1cm},clip]{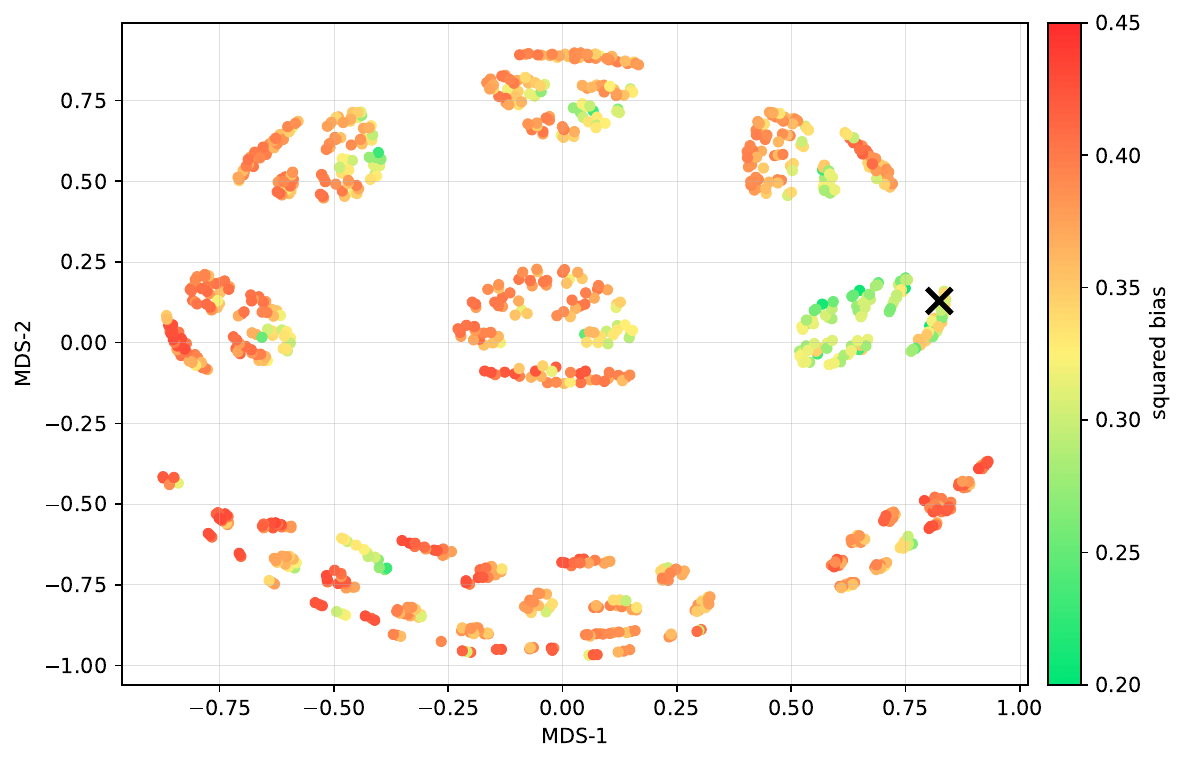}
\end{tabular}
\end{minipage}
\hspace{-0.01\textwidth}
\begin{minipage}[c]{0.08\textwidth}
\centering
\includegraphics[
    height=0.452\textheight,
    trim={17.7cm 1cm 0.15cm 0.05cm},
    clip
]{0_Figures/6-class-bias-right-cluster.pdf}
\end{minipage}
\caption{Squared bias across the 6-class ndscape embedding for four generating structures from the 1000 sparse-feature setting. The cross marks the generating ND. Low-bias regions move with the generator rather than staying fixed in one part of the map. Interactive plots are available at \url{https://res-lucid.github.io/ndscape-plots/bias-variance/}}
\label{fig:6-class-bias-four-generators}
\end{figure*}

\begin{figure*}[t]
\centering
\setlength{\tabcolsep}{0pt}
\renewcommand{\arraystretch}{0}
\begin{minipage}[c]{0.91\textwidth}
\centering
\begin{tabular}{@{}cc@{}}
\includegraphics[height=0.225\textheight,trim={1.0cm 0.8cm 2.6cm 0.1cm},clip]{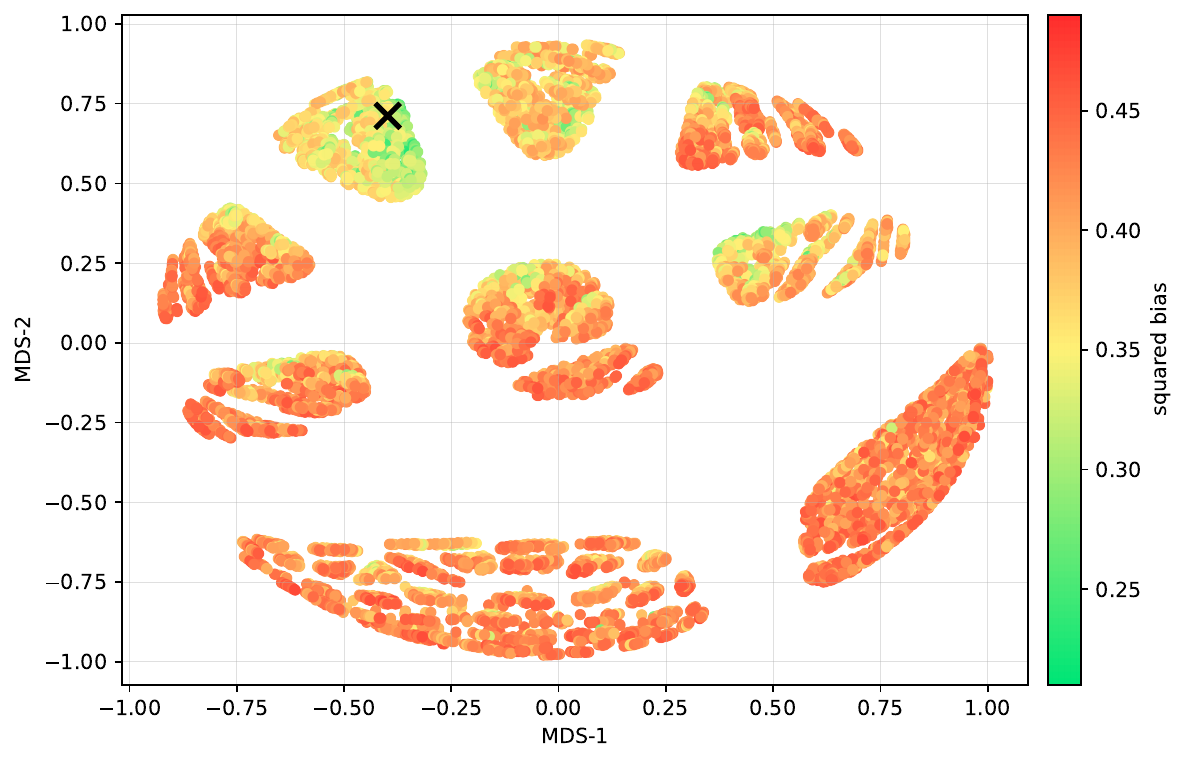}
&
\includegraphics[height=0.225\textheight,trim={1.9cm 0.8cm 2.6cm 0.1cm},clip]{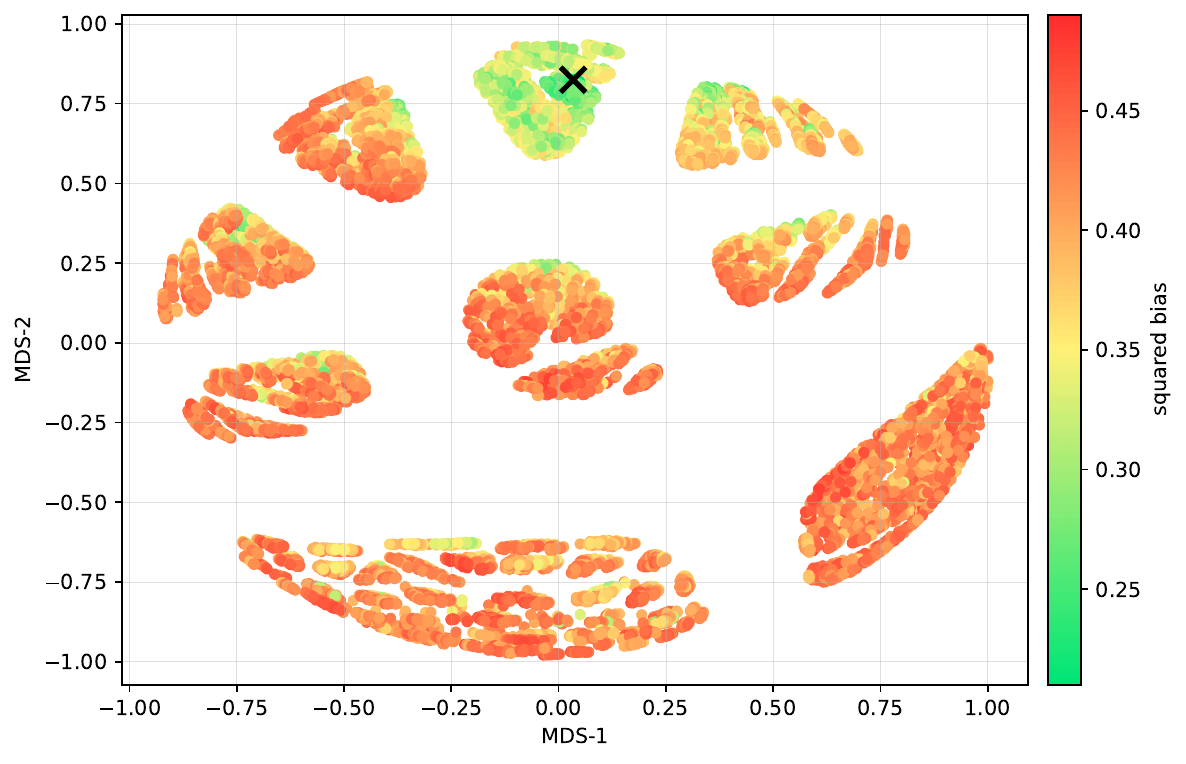}
\\[-0.18cm]
\includegraphics[height=0.225\textheight,trim={1.0cm 0.8cm 2.6cm 0.1cm},clip]{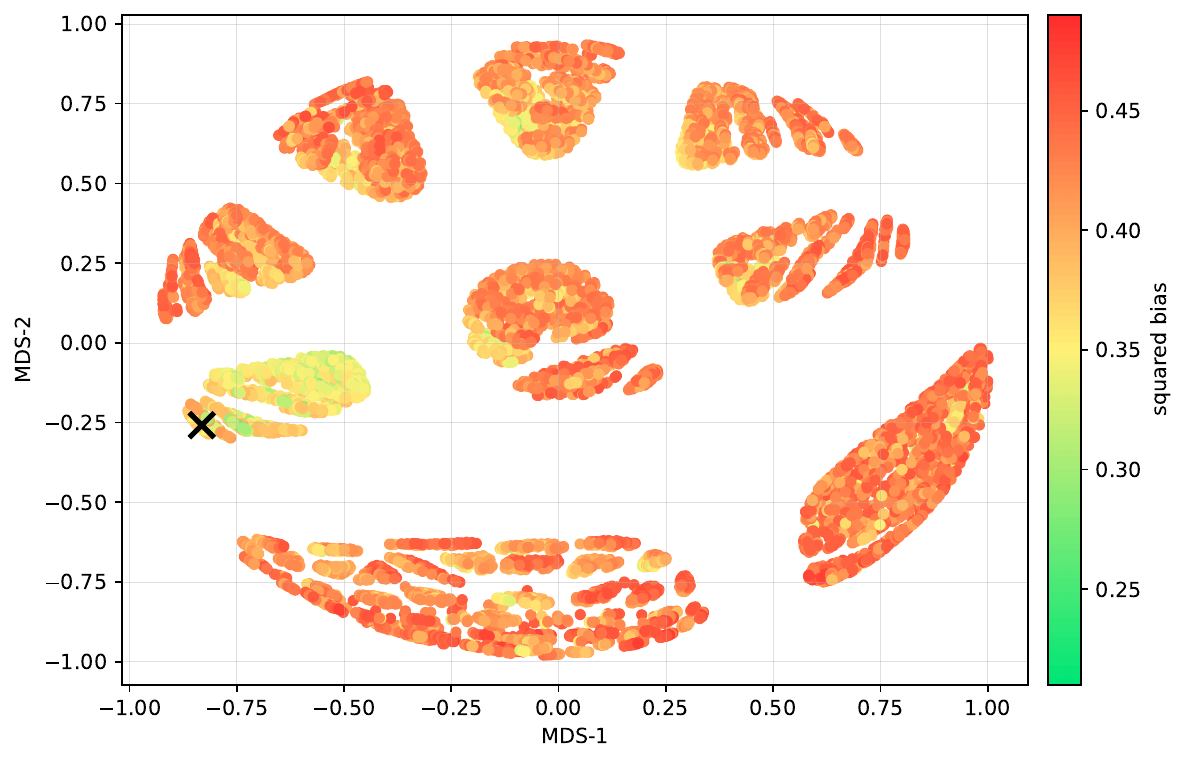}
&
\includegraphics[height=0.225\textheight,trim={1.9cm 0.8cm 2.6cm 0.1cm},clip]{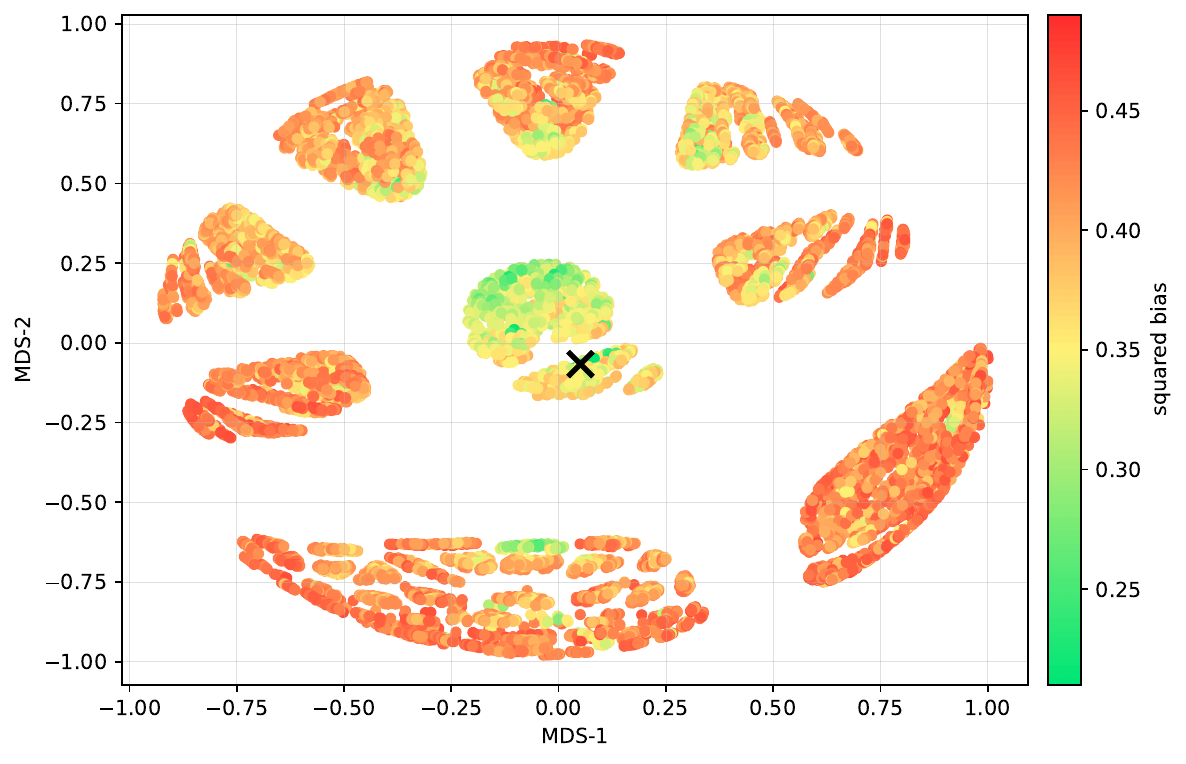}
\end{tabular}
\end{minipage}
\hspace{-0.01\textwidth}
\begin{minipage}[c]{0.08\textwidth}
\centering
\includegraphics[
    height=0.452\textheight,
    trim={17.7cm 1cm 0.15cm 0.05cm},
    clip
]{0_Figures/7-bias-center-cluster.pdf}
\end{minipage}
\caption{Squared bias across the 7-class ndscape embedding for four generating structures. Local concentration of low-bias NDs around the generating ND persists in the larger space}
\label{fig:seven-class-bias-plots}
\end{figure*}

Variance behaves differently from squared bias and squared probability loss. It is spatially organised more strongly than either, with Moran's $I$ rising from $0.541$ at 6 classes to $0.833$ at 10 classes, but its enrichment near the generator is weak. At 6 classes, the nearest $10\%$ of structures hold only $26.1\%$ of the lowest-variance decile, and at 7 classes, the AUC of distance to the generator is $0.502$, no better than chance. Instead, mean root-to-leaf depth organises variance. Longer paths multiply more conditional probability estimates from root to leaf, and each estimate carries its own sampling variation. Deeper NDs are therefore more variable, and more balanced NDs are more stable. The bivariate Moran's $I$ between variance and depth captures this, rising from $0.589$ at 6 classes to $0.883$ at 10 classes and exceeding the univariate $I$ for bias and loss at every class count. The non-spatial Spearman correlation between a tree's variance and its mean depth is $0.72$ for 6 classes and increases to $0.88$ for 10 classes. Mean path length is a stronger and more stable organiser of variance than the generating ND is of bias or loss. Figures~\ref{fig:variance-depth-6-class} and~\ref{fig:variance-depth-7-class} show this, with variance rising in the deeper regions of the map regardless of the generating ND.
 
\begin{figure}[t]
\centering
\makebox[\textwidth][c]{%
\begin{minipage}{0.52\textwidth}
    \centering
    \includegraphics[width=\linewidth]{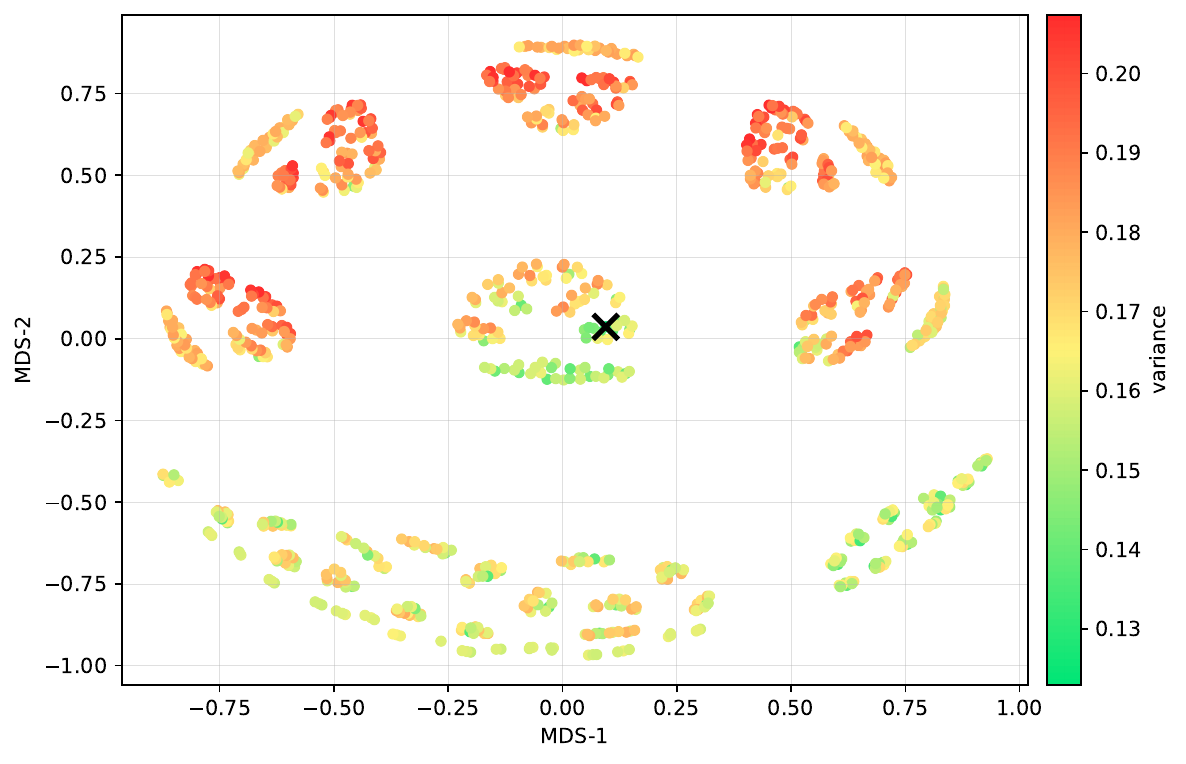}
    \caption*{(a) Variance}
\end{minipage}
\hspace{-0.02\textwidth}
\begin{minipage}{0.52\textwidth}
    \centering
    \includegraphics[width=\linewidth]{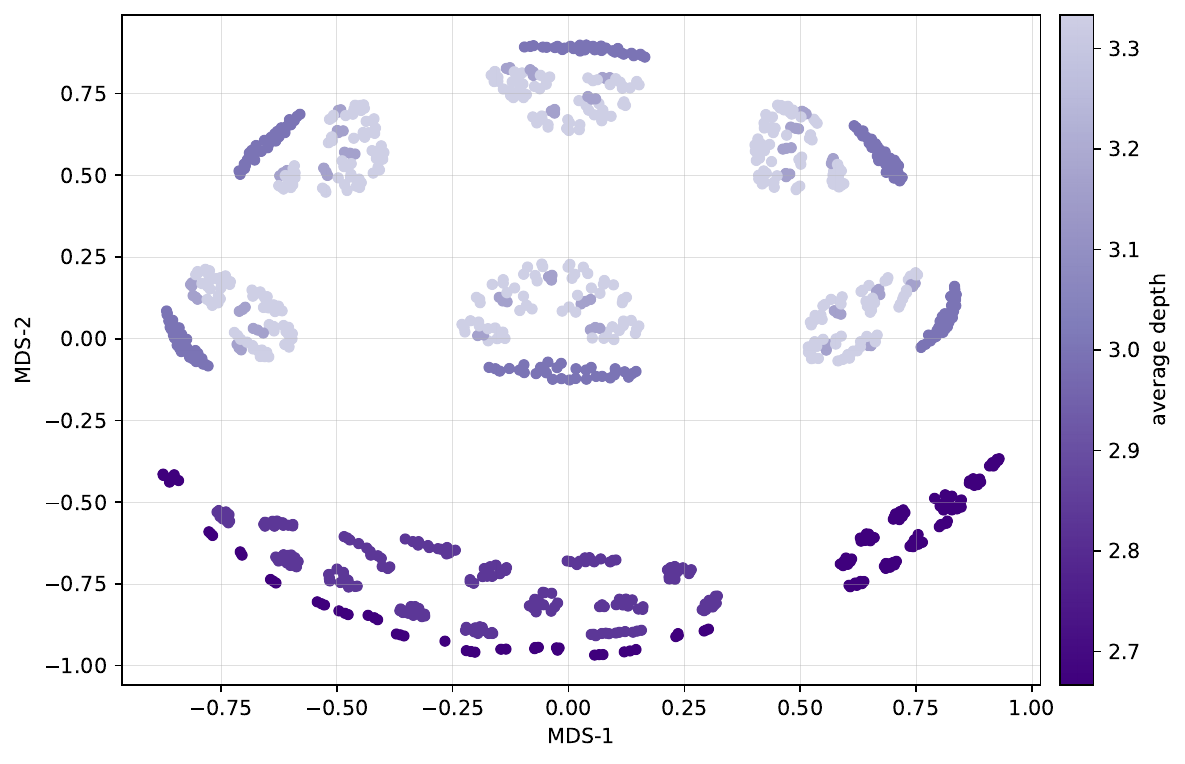}
    \caption*{(b) Average depth}
\end{minipage}
}
\caption{(a) Variance and (b) average depth on the 6-class ndscape embedding. Variance is spatially clustered and follows the depth pattern more strongly than squared bias or squared probability loss}
\label{fig:variance-depth-6-class}
\end{figure}

\begin{figure}[t]
\centering
\makebox[\textwidth][c]{%
\begin{minipage}{0.52\textwidth}
    \centering
    \includegraphics[width=\linewidth]{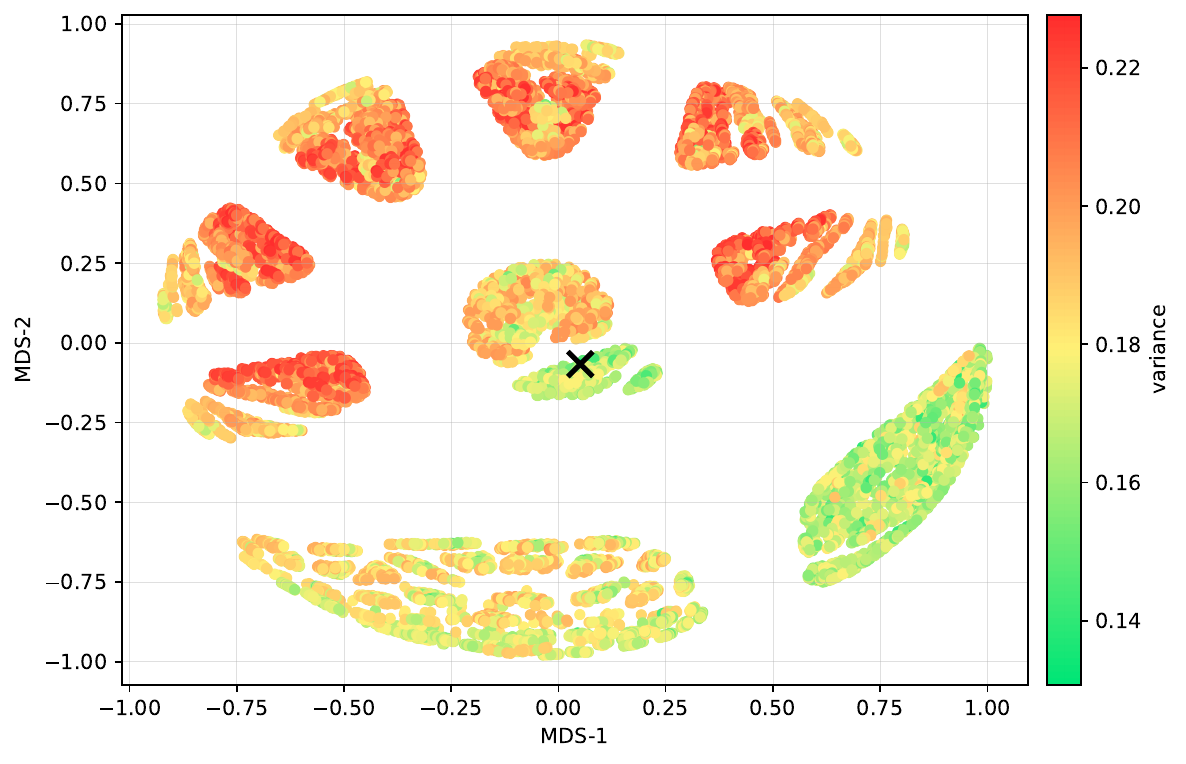}
    \caption*{(a) Variance}
\end{minipage}
\hspace{-0.02\textwidth}
\begin{minipage}{0.52\textwidth}
    \centering
    \includegraphics[width=\linewidth]{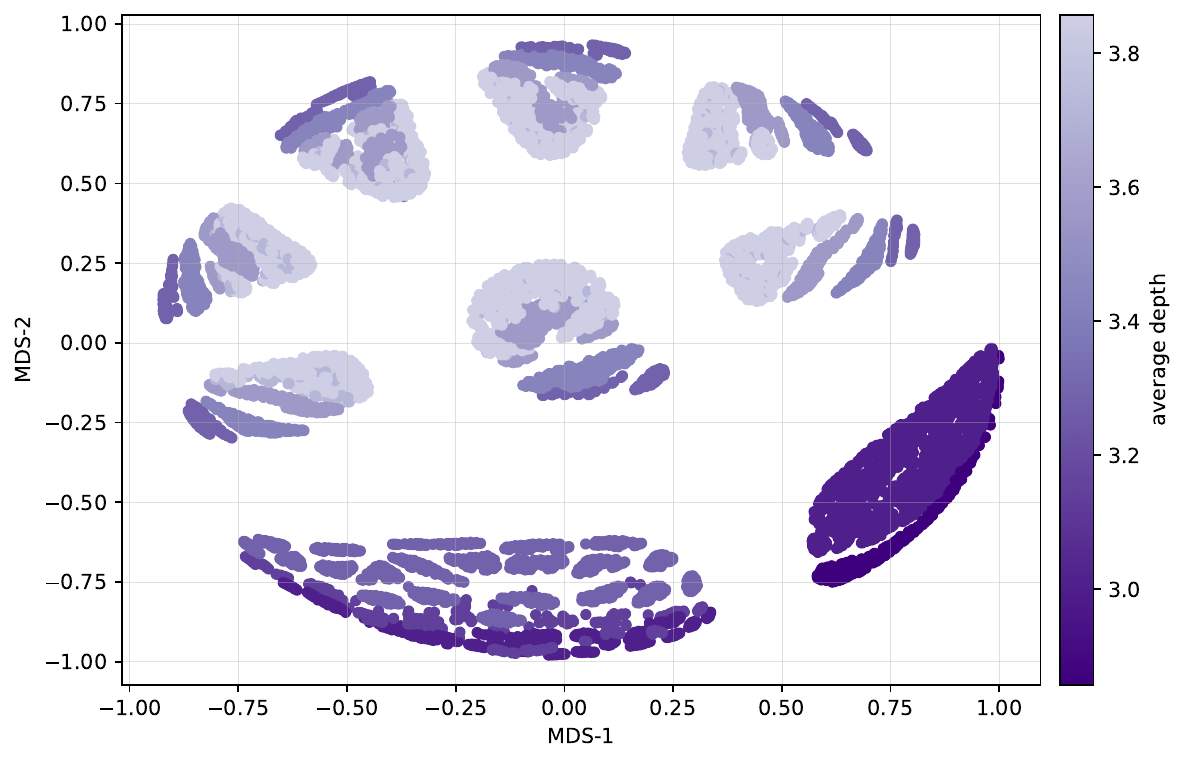}
    \caption*{(b) Average depth}
\end{minipage}
}
\caption{(a) Variance and (b) average depth on the 7-class ndscape embedding. The variance--depth association seen at 6 classes persists in the larger space of 10{,}395 structures}
\label{fig:variance-depth-7-class}
\end{figure}

The depth organisation of variance is specific to the logistic node model. With a decision tree at each node, squared bias and squared probability loss stay spatially organised, while variance is nearly flat and tracks squared bias rather than depth. Appendix~\ref{appendix:dt} gives the full comparison across three feature settings.

Interactive versions of all plots are available at \url{https://res-lucid.github.io/ndscape-plots/}.

\FloatBarrier
\section{Experiments on real data}\label{sec:real_data}

\subsection{Experimental procedure}\label{sec:real-protocol}

As the simulations use known generating NDs, squared bias and variance can be estimated directly. In real datasets, the true probabilities are unknown, so we cannot estimate squared bias or squared probability loss directly. Instead, we test whether held-out accuracy, log loss, and prediction variance are spatially organised in the ndscape embedding. We estimate prediction variance as the bootstrap variance of the fitted class probabilities across 100 resampled training sets, averaged over held-out points and classes.

The base learner is logistic regression as in Section~\ref{sec:simulation_design}, and Moran's $I$ is the same as in Section~\ref{sec:spatial_summaries}. Clustering is performed using $k$-means based on the silhouette coefficient, as in the simulation.

We run six construction heuristics (RPND, BoK, BBoK, ACND, CBND, NDC) on each dataset, using the same train and test splits throughout. We use fixed splits because the unit compared here is the ND structure, and changing the split would conflate resampling variation with structural variation. Each heuristic returns one ND. We tokenise that ND the same way as every other structure and place it at its embedding coordinate, so a heuristic's ND sits on the map alongside the evaluated trees and falls inside one of the clusters. BoK and BBoK each draw 10 candidate NDs and select the one with the lowest training error, a budget chosen to keep their runtime in line with the other heuristics.

The six datasets are drawn from the UCI Machine Learning Repository \citep{kelly2024uci} and span 6 to 19 classes (Table~\ref{tab:realdata-datasets}). Glass Identification \citep{glass_identification} and Steel Plates Faults \citep{steel_plates_faults} sit in the range where every ND can be enumerated. Mice Protein \citep{mice_protein}, Urban Land Cover \citep{urban_land_cover}, Pen-Based Digits \citep{pen-based_recognition_of_handwritten_digits_81}, and Soybean \citep{soybean_large} fall outside it, so we sample 15{,}000 structures per dataset using the sampler from Appendix~\ref{appendix:sampling}. For the four sampled datasets, we write ``best evaluated tree'' rather than ``best tree''.

\begin{table}[t]\centering
\caption{The six real datasets. All NDs are evaluated at 6 and 7 classes, and a 15{,}000-tree sample is used for the 8-, 9-, 10-, and 19-class datasets}
\label{tab:realdata-datasets}\small\renewcommand{\arraystretch}{1.28}
\setlength{\tabcolsep}{0pt}
\begin{tabular*}{\textwidth}{@{\extracolsep{\fill}}lccccc@{}}
\toprule
Dataset & Classes & Samples & Features & Split & Trees evaluated \\
\midrule
Glass Identification & 6  & 214    & 9   & 70/30 stratified & 945 \\
Steel Plates Faults & 7  & 1{,}941 & 27  & 70/30 stratified & 10{,}395 \\
Mice Protein & 8  & 1{,}080 & 77  & 70/30 stratified & 15{,}000 \\
Urban Land Cover & 9  & 675    & 147 & UCI predefined   & 15{,}000 \\
Pen-Based Digits  & 10 & 10{,}992 & 16  & 70/30 stratified & 15{,}000 \\
Soybean & 19 & 307 & 35  & 70/30 stratified & 15{,}000 \\
\bottomrule
\end{tabular*}
\end{table}

\subsection{Spatial organisation and heuristic placement}\label{sec:real-spatial}

Table~\ref{tab:realdata-spatial} reports Moran's $I$ for held-out accuracy, prediction variance, and log loss. Every entry is positive at $p \le 0.001$ across all three measures and all six datasets. An ND's held-out performance is therefore partly predicted by the performance of its neighbours.

Accuracy autocorrelation ranges from $0.165$ on Urban Land Cover to $0.448$ on Steel Plates Faults, and log loss shows a similar range. Prediction variance is also positively autocorrelated, with Moran's $I$ ranging from $0.165$ for Urban Land Cover to $0.461$ for Glass Identification. The variance proxy tracks accuracy more than depth. The bivariate Moran's $I$ between depth and variance is positive on the enumerated Glass and Steel sets ($0.12$ and $0.26$), matching the simulation, but negative on the four sampled sets, whose more accurate NDs tend to be deeper and, under resampling, more stable. With accuracy held fixed, the partial Spearman correlation between depth and variance is positive on Glass, Steel, and Pen-Based Digits, near zero on Mice Protein and Urban Land Cover, and negative only on Soybean.

Trust@20 (the trustworthiness of the 2D embedding at 20 nearest neighbours) is at least $0.98$ everywhere, meaning the two-dimensional map preserves the nearest-neighbour structure of the cosine distance. Urban Land Cover is weakest across all three measures, consistent with its flat cluster structure (Appendix~\ref{appendix:urban_land}). The Soybean dataset extends the embedding beyond the class counts considered in simulation, and even with 19 classes, the embedding still organises all three measures ($0.239$ to $0.364$). This suggests that the representation and sampling procedure remain informative beyond 10 classes.

\begin{table}[t]\centering
\caption{Spatial statistics across the six datasets. Under no spatial organisation, each Moran's $I$ has an expected value near zero. Trust@20 is the trustworthiness of the 2D embedding at 20 nearest neighbours. All Moran's $I$ have $p \le 0.001$}
\label{tab:realdata-spatial}\small\renewcommand{\arraystretch}{1.28}
\begin{tabular*}{\textwidth}{@{\extracolsep{\fill}}lccccc@{}}
\toprule
Dataset & Classes & Trust@20 & $I$ (accuracy) & $I$ (variance) & $I$ (log loss) \\
\midrule
Glass Identification         & 6  & 0.982 & 0.362 & 0.461 & 0.380 \\
Steel Plates Faults          & 7  & 0.988 & 0.448 & 0.332 & 0.421 \\
Mice Protein                 & 8  & 0.985 & 0.407 & 0.413 & 0.402 \\
Urban Land Cover             & 9  & 0.988 & 0.165 & 0.165 & 0.191 \\
Pen-Based Digits             & 10 & 0.985 & 0.323 & 0.179 & 0.347 \\
Soybean                      & 19 & 0.986 & 0.254 & 0.239 & 0.364 \\
\bottomrule
\end{tabular*}
\end{table}

We also map the structures selected by heuristics into this space. Under the fixed split and stated search budget, the heuristics do not find the best individual NDs, nor do they reach the best-performing cluster on any of the six datasets (Table~\ref{tab:realdata-heuristic}). The best heuristic region is still below the top region, with gaps ranging from 0.4 percentage points on Soybean to 5.9 on Glass.

\begin{table}[t]\centering
\caption{Embedding placement of construction heuristics on the real-data maps under fixed splits and stated search budgets.  The best heuristic is the highest-accuracy ND selected by any of the six heuristics. The best region reached is the embedding cluster containing that best heuristic ND, with rank 1 the highest-mean cluster. Region gap is the mean accuracy of the best cluster minus the mean accuracy of that heuristic's cluster (pp = percentage points)}
\label{tab:realdata-heuristic}\small\renewcommand{\arraystretch}{1.25}
\begin{tabular*}{\textwidth}{@{\extracolsep{\fill}}lclcc@{}}
\toprule
Dataset & $M$ & Best heuristic & Best region reached & Region gap (pp) \\
\midrule
Glass Identification & 6 & BBoK & 6 of 9 & 5.9 \\
Steel Plates Faults & 7 & BoK & 4 of 10 & 0.7 \\
Mice Protein & 8 & NDC & 19 of 19 & 4.1 \\
Urban Land Cover & 9 & NDC & 12 of 12 & 3.0 \\
Pen-Based Digits & 10 & BoK & 4 of 17 & 1.5 \\
Soybean & 19 & NDC & 8 of 27 & 0.4 \\
\bottomrule
\end{tabular*}
\end{table}

\subsection{Reading the map on Glass Identification}\label{sec:real-glass}

The Glass Identification dataset has 214 samples, 9 continuous chemical measurements, and 6 classes. Across 945 NDs, held-out accuracy ranges from $36.9\%$ to $63.1\%$, with a mean of $52.3\%$, so the structure alone moves accuracy by more than 26 points.

Figure~\ref{fig:glass} colours the embedding by held-out accuracy. The embedding contains 9 clusters, 6 with only NDs that isolate a single glass type at the root. Among these singleton-root families, the isolated class largely predicts the cluster's accuracy. Isolating ``building (float)" glass gives the best cluster, with a mean accuracy of $57.7\%$, and contains the three tied highest-accuracy NDs at $63.1\%$. Isolating ``building (non float)" glass gives the worst cluster, with a mean accuracy of $43.8\%$. The $13.9$-point gap between these two clusters is driven by only the top split.

The two building glass clusters are also the most compact. The ``building (non float)" family has a spread of only $\pm 0.9$ points, and the ``building (float)" family has a spread of $\pm 3.0$ points. The other singleton-root clusters have much wider spreads, from $\pm 6.0$ to $\pm 6.4$ points. In the other clusters, the root appears to have less effect, as the mean accuracies differ by less than their within-family spread. The map therefore gives one readable instruction: separate ``building (float)" from the rest at the root, and avoid separating ``building (non float)" first. An inspection of the features shows that ``building (float)" has higher \texttt{Mg} and \texttt{Al} than the other classes, while ``building (non float)" is less clearly separated by the available chemical variables.

However, none of the six heuristics selects an ND in this ``building (float)" cluster. Their NDs land in mixed-root clusters ranked 6th and 8th among the nine clusters. BBoK finds the best heuristic ND at $58.5\%$, but it comes from a region that averages only $51.8\%$. The embedding simplifies the choice of ND and gives insight into why a structure with ``building (float)" separated ought to be selected. 

\begin{figure}\centering
\includegraphics[width=1\linewidth]{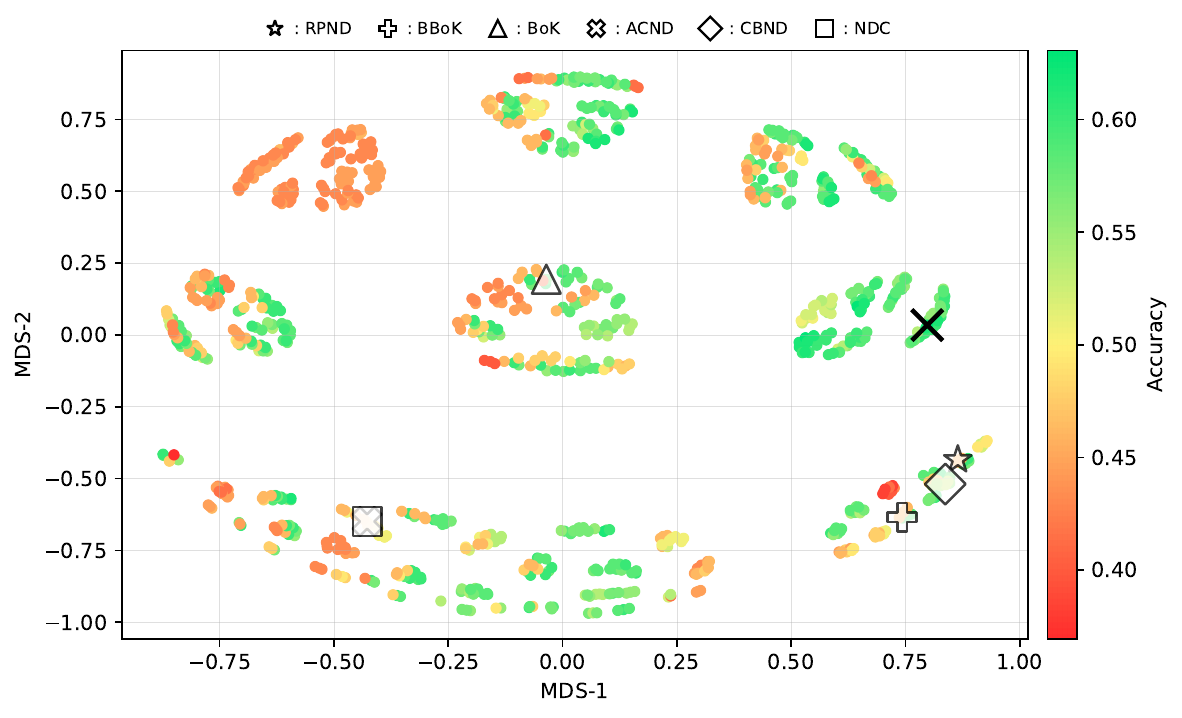}
\caption{Embedding of all 945 NDs for the Glass Identification dataset, coloured by held-out accuracy. Each point is an ND, positioned by MDS on the cosine distance between topology encodings. Marker accuracies are RPND ($46.2\%$), BoK ($56.9\%$), BBoK ($58.5\%$), ACND ($49.2\%$), CBND ($52.3\%$), and NDC ($49.2\%$). The solid cross marks the highest-accuracy ND ($63.1\%$)}
\label{fig:glass}
\end{figure}

\section{Discussion}

Across both simulated and real datasets, structurally similar NDs perform similarly, and the topology-only embedding makes this visible. Low squared-bias and low squared-loss structures concentrate near the generating structure, while variance follows a different pattern, driven mainly by tree depth. This section explains what the map shows, how real-world data extends the simulations, and how a practitioner might use the embedding.

\subsection{Where do the best and worst NDs sit on the map?}

Structures with low squared bias and low squared probability loss concentrate near the generating ND. At 6 and 7 classes, the structures nearest the generator capture a large proportion of the best ND structures. The layout for a class count is fixed, yet the low-bias and low-loss region moves with the generator (see Figures~\ref{fig:6-class-bias-four-generators} and~\ref{fig:seven-class-bias-plots}). In the real datasets, the highest-accuracy NDs sit among other high-performing structures.

Moran's $I$ shows the organisation reaches beyond the best ND. For squared bias, Moran's $I$ is positive across all class counts, from $0.342$ to $0.473$, and so is accuracy across all six real datasets, strongest on Steel Plates Faults ($0.448$) and weakest on Urban Land Cover ($0.165$). The consistent pattern is that good trees neighbour good trees, and poor trees neighbour poor trees.

\subsection{Why does this fixed embedding explain predictive behaviour?}

NDs with similar node partitions route probability mass in similar ways, and the top-level split has the greatest effect because every class probability estimate must pass through it. The tokeniser reflects this by weighting node partitions according to depth. The ablation against an unweighted alternative shows the value of these weights (see bag-of-words (BoW) in Tables~\ref{tab:ablation-1000-sparse}, \ref{tab:ablation-appendix-6}, and~\ref{tab:ablation-appendix-7}).

\subsection{How does variance behave?}

Variance is organised mainly by depth, with shorter root-to-leaf paths corresponding to lower variance, so flatter, more balanced trees tend to be more stable. We attribute this to the way an ND produces probability estimates. The probability that a new instance belongs to a class is the product of the left-or-right branch probabilities along the path from the root to its leaf. A shorter path multiplies fewer estimates, so it accumulates less of the sampling variation that each estimate carries. Class-balanced ND (CBND) produces balanced trees by design to reduce runtime \citep{Dong2005}, and our results suggest this also yields more stable trees on average.

Two things can weaken this depth organisation. The first is the node model. With a decision tree at each node rather than logistic regression, the variance surface is nearly flat, variance follows squared bias, and the depth association disappears (Appendix~\ref{appendix:dt}). The second is a confound in the four sampled real datasets, where deeper trees tend to be more accurate and, under resampling, more stable. That correlation turns the raw depth-variance association negative even though the base learner is unchanged, and holding accuracy fixed removes most of the reversal (Section~\ref{sec:real-spatial}). In every case variance stays spatially organised, even where its tie to depth slackens, so the tokenisation still captures how variance propagates through the ND.

\subsection{What do the real-data results add?}

The simulation tests the embedding under a known generating process, whereas the six real datasets test it under an unknown generating process. As in the simulations, spatial autocorrelation is positive for every measure on every dataset. This is most surprising on the 19-class Soybean dataset, where only 15{,}000 structures are sampled from a possible $2.2 \times 10^{20}$. Even at this class count, the map continues to organise the structure space.

The heuristics can be assessed within the embedding by comparing each selected tree with its surrounding region. Some selections come close to the best evaluated tree, including BoK on Pen-Based Digits, which is within $0.7$ points. The weakest selection, RPND on Glass, falls $16.9$ points short. Even so, no heuristic lands in the best-performing region on any dataset.

\subsection{How should the map be used?}

A practitioner should build the map once for the class count of interest, score the candidate structures, and colour the result by held-out accuracy, log loss, prediction variance, or any other metric of interest. Since neighbouring NDs behave alike, a sample is enough to locate the region of strong NDs even in spaces too large to fully enumerate. Rather than a single recommended ND, the map shows the effect of class groupings at different levels of the structure. For instance, on the Glass dataset, the embedding reduces a choice among 945 structures to a single design decision: to separate ``building (float)'' glass from the rest at the root. That decision states a property of the data rather than of any single model, and it can be checked against the features. 

Existing heuristics return an ND structure without accounting for how it compares with alternatives or which class groupings drive its performance. Placing the heuristics on the map adds that context by showing the region around each returned tree.

Interactive plots, such as those at \url{https://res-lucid.github.io/ndscape-plots/}, allow practitioners to inspect individual trees and analyse how particular node partitions suit their data. The same workflow can be applied to new datasets using the \texttt{ndscape} Python package, available at \url{https://github.com/res-lucid/ndscape} \citep{ndscape}. 

\section{Limitations and future work}\label{sec:future_work}

A limitation of the present study is that the main experiments use only logistic regression as the node model. The embedding is topology-only, but the performance overlaid on it depends on the fitted learner. The decision-tree sensitivity analysis in Appendix~\ref{appendix:dt} shows that the embedding still organises performance under a different node model, and that the depth association disappears once variance becomes flat. A broader comparison across node models and data-generating mechanisms would further test this, but the main simulation already fits 79.2 million NDs, which puts a full learner sweep beyond the present scope.

The real-data study also uses one held-out split per dataset. This was chosen to isolate the effect of tree structure, since every candidate ND is trained and evaluated on the same split. The results should not be read as cross-validated estimates of the absolute accuracy of any construction heuristic. They support the claim that, under a fixed data realisation, empirical performance is spatially organised and heuristic outputs can be placed relative to high-performing regions. Repeated splits would be the natural next step for benchmarking an embedding-guided search rule.

Another direction is to develop a search heuristic based on our findings. The real-data results show that the heuristics do not reach the best-performing regions. An embedding-guided search could identify the best cluster from a sample and then explore it in depth. We do not consider such an algorithm here. This paper maps the ND space and shows where heuristics fall short, while the development and benchmarking of an embedding-guided search are left to future work.

\section{Conclusion}\label{sec:conclusion}
This paper introduces a topology-only embedding of the ND structure space. Built from depth-weighted structural tokens before any model is fitted, the embedding places structurally similar trees close together and turns an abstract set of node partitions into a readable map. It reframes the choice of an ND as a comparison of class groupings rather than a single construction step, and supports what is, to our knowledge, the first bias-variance analysis across the space of ND structures.

In simulations with 6 to 10 classes and a known generating ND, squared bias, variance, and squared probability loss are each spatially organised, so a structure performs much like its neighbours. Squared bias and squared probability loss concentrate near the generating ND, whereas variance is organised by depth, so flatter, more balanced trees tend to be more stable under logistic regression, though not necessarily more accurate. Replacing the logistic node model with a decision tree preserves the bias and loss organisation but flattens the variance surface, indicating that depth-driven stability is a property of the logistic setting rather than of the embedding. The same pattern appears on six UCI datasets spanning 6 to 19 classes, where held-out accuracy, log loss, and bootstrap prediction variance are organised in the same way, and it holds even on the 19-class problem, whose structure space is far too large to enumerate.

Mapping the six standard construction heuristics (RPND, BoK, BBoK, ACND, CBND, NDC) onto this space shows that, although several select strong individual NDs, none reaches the best-performing region on any of the six datasets. The embedding, therefore, does more than identify a single good structure. It places every heuristic and every candidate on a common map, exposes the upper-level class groupings that separate high-performing and stable regions from the rest, and lets a practitioner check that structural choice against the features. With ndscape, researchers and practitioners can see where a structure sits among its alternatives and why it performs as it does.

\section*{Acknowledgements}

This research was supported by use of the Nectar Research Cloud, a collaborative Australian research platform supported by the National Collaborative Research Infrastructure Strategy (NCRIS)-funded Australian Research Data Commons (ARDC). This work was supported by resources provided by the Pawsey Supercomputing Research Centre with funding from the Australian Government and the Government of Western Australia.

\section*{Statements and Declarations} 
\subsection*{Funding} 
This research was supported by an Australian Government Research Training Program (RTP) Scholarship, \href{https://doi.org/10.82133/C42F-K220}{https://doi.org/10.82133/C42F-K220}. 

\subsection*{Competing interests} 
The authors have no competing interests to declare that are relevant to the content of this article. 

\subsection*{Ethics approval and consent to participate}
Not applicable.

\subsection*{Data and code availability} 
The six real datasets analysed in this study are publicly available from the UCI Machine Learning Repository \citep{kelly2024uci}: Glass Identification \citep{glass_identification}, Steel Plates Faults \citep{steel_plates_faults}, Mice Protein Expression \citep{mice_protein}, Urban Land Cover \citep{urban_land_cover}, Pen-Based Recognition of Handwritten Digits \citep{pen-based_recognition_of_handwritten_digits_81}, and Soybean (Large) \citep{soybean_large}. No new real-world datasets were collected. The simulated datasets are generated by the data-generating process described in Section~\ref{sec:simulation_design} and Appendix~\ref{appendix:feature_gen}, and can be regenerated from the released code. The code used to construct the embeddings, run the simulations, analyse the results, and generate the interactive maps is available in the open-source \texttt{ndscape} package at \url{https://github.com/res-lucid/ndscape} \citep{ndscape}. The interactive plots are available at \url{https://res-lucid.github.io/ndscape-plots/}.

\bibliographystyle{unsrtnat}

\begin{thebibliography}{34}
\providecommand{\natexlab}[1]{#1}
\providecommand{\url}[1]{\texttt{#1}}
\expandafter\ifx\csname urlstyle\endcsname\relax
  \providecommand{\doi}[1]{doi: #1}\else
  \providecommand{\doi}{doi: \begingroup \urlstyle{rm}\Url}\fi

\bibitem[Frank and Kramer(2004)]{FrankEibe2004Eond}
Eibe Frank and Stefan Kramer.
\newblock Ensembles of nested dichotomies for multi-class problems.
\newblock In \emph{Proceedings of the Twenty-First International Conference on
  Machine Learning}, pages 305--312, 2004.
\newblock \doi{10.1145/1015330.1015363}.
\newblock \url{https://doi.org/10.1145/1015330.1015363}.

\bibitem[Fdez-Sánchez et~al.(2021)Fdez-Sánchez, Pascual-Triana, Fernández,
  and Herrera]{FdezSanchez2021}
J.A. Fdez-Sánchez, J.D. Pascual-Triana, A.~Fernández, and F.~Herrera.
\newblock Learning interpretable multi-class models by means of hierarchical
  decomposition: Threshold control for nested dichotomies.
\newblock \emph{Neurocomputing}, 463:\penalty0 514--524, 2021.
\newblock \doi{10.1016/j.neucom.2021.07.097}.

\bibitem[Rodr{\'\i}guez et~al.(2010)Rodr{\'\i}guez, Garc{\'\i}a-Osorio, and
  Maudes]{rodriguez2010forests}
Juan~J Rodr{\'\i}guez, C{\'e}sar Garc{\'\i}a-Osorio, and Jes{\'u}s Maudes.
\newblock Forests of nested dichotomies.
\newblock \emph{Pattern Recognition Letters}, 31\penalty0 (2):\penalty0
  125--132, 2010.
\newblock \doi{10.1016/j.patrec.2009.09.015}.

\bibitem[Wever et~al.(2018)Wever, Mohr, and H{\"u}llermeier]{Wever2018}
Marcel Wever, Felix Mohr, and Eyke H{\"u}llermeier.
\newblock Ensembles of evolved nested dichotomies for classification.
\newblock In \emph{Proceedings of the Genetic and Evolutionary Computation
  Conference}, GECCO '18, pages 561--568, New York, NY, USA, 2018. Association
  for Computing Machinery.
\newblock \doi{10.1145/3205455.3205562}.
\newblock \url{https://doi.org/10.1145/3205455.3205562}.

\bibitem[Wever et~al.(2023)Wever, {\"O}zdogan, and
  H{\"u}llermeier]{evolutionary_ensembling2023}
Marcel Wever, Miran {\"O}zdogan, and Eyke H{\"u}llermeier.
\newblock Cooperative co-evolution for ensembles of nested dichotomies for
  multi-class classification.
\newblock In \emph{Proceedings of the Genetic and Evolutionary Computation
  Conference}, GECCO '23, page 597–605, New York, NY, USA, 2023. Association
  for Computing Machinery.
\newblock \doi{10.1145/3583131.3590457}.
\newblock \url{https://doi.org/10.1145/3583131.3590457}.

\bibitem[Dong et~al.(2005)Dong, Frank, and Kramer]{Dong2005}
Lin Dong, Eibe Frank, and Stefan Kramer.
\newblock Ensembles of balanced nested dichotomies for multi-class problems.
\newblock In Al{\'i}pio~M{\'a}rio Jorge, Lu{\'i}s Torgo, Pavel Brazdil, Rui
  Camacho, and Jo{\~a}o Gama, editors, \emph{Knowledge Discovery in Databases:
  PKDD 2005}, volume 3721 of \emph{Lecture Notes in Computer Science}, pages
  84--95. Springer Berlin Heidelberg, Berlin, Heidelberg, 2005.
\newblock \doi{10.1007/11564126\_13}.
\newblock \url{https://doi.org/10.1007/11564126_13}.

\bibitem[Duarte-Villase{\~n}or et~al.(2012)Duarte-Villase{\~n}or,
  Carrasco-Ochoa, Mart{\'\i}nez-Trinidad, and Flores-Garrido]{duarte2012nested}
Miriam~M{\'o}nica Duarte-Villase{\~n}or, Jes{\'u}s~Ariel Carrasco-Ochoa,
  Jos{\'e}~Francisco Mart{\'\i}nez-Trinidad, and Marisol Flores-Garrido.
\newblock Nested dichotomies based on clustering.
\newblock In Luis Alvarez, Marta Mejail, Luis Gomez, and Julio Jacobo, editors,
  \emph{Progress in Pattern Recognition, Image Analysis, Computer Vision, and
  Applications}, volume 7441 of \emph{Lecture Notes in Computer Science}, pages
  162--169, Berlin, Heidelberg, 2012. Springer.
\newblock \doi{10.1007/978-3-642-33275-3\_20}.
\newblock \url{https://doi.org/10.1007/978-3-642-33275-3_20}.

\bibitem[Leathart et~al.(2016)Leathart, Pfahringer, and Frank]{Leathart2016}
Tim Leathart, Bernhard Pfahringer, and Eibe Frank.
\newblock Building ensembles of adaptive nested dichotomies with random-pair
  selection.
\newblock In Paolo Frasconi, Niels Landwehr, Giuseppe Manco, and Jilles
  Vreeken, editors, \emph{Machine Learning and Knowledge Discovery in
  Databases}, volume 9852 of \emph{Lecture Notes in Artificial Intelligence},
  pages 179--194, Cham, 2016. Springer.
\newblock \doi{10.1007/978-3-319-46227-1\_12}.
\newblock \url{https://doi.org/10.1007/978-3-319-46227-1_12}.

\bibitem[Yang et~al.(2017)Yang, Destercke, and Masson]{Yang2017}
Gen Yang, Sébastien Destercke, and Marie-H{\'e}l{\`e}ne Masson.
\newblock Cautious classification with nested dichotomies and imprecise
  probabilities.
\newblock \emph{Soft Computing}, 21\penalty0 (24):\penalty0 7447--7462, 2017.
\newblock \doi{10.1007/s00500-016-2287-7}.

\bibitem[Melnikov and H{\"u}llermeier(2018)]{Melnikov2018}
Vitalik Melnikov and Eyke H{\"u}llermeier.
\newblock On the effectiveness of heuristics for learning nested dichotomies:
  an empirical analysis.
\newblock \emph{Machine Learning}, 107:\penalty0 1537--1560, 2018.
\newblock \doi{10.1007/s10994-018-5733-1}.

\bibitem[Dix-Matthews(2026)]{ndscape}
Maxwell Dix-Matthews.
\newblock ndscape.
\newblock GitHub repository, 2026.
\newblock URL \url{https://github.com/res-lucid/ndscape}.

\bibitem[Geman et~al.(1992)Geman, Bienenstock, and Doursat]{geman1992}
Stuart Geman, Elie Bienenstock, and René Doursat.
\newblock Neural networks and the bias/variance dilemma.
\newblock \emph{Neural Computation}, 4:\penalty0 1--58, 01 1992.
\newblock \doi{10.1162/neco.1992.4.1.1}.

\bibitem[Kohavi(1995)]{kohavi1995study}
Ron Kohavi.
\newblock A study of cross-validation and bootstrap for accuracy estimation and
  model selection.
\newblock In \emph{Proceedings of the 14th International Joint Conference on
  Artificial Intelligence}, volume~2 of \emph{IJCAI'95}, page 1137–1143, San
  Francisco, CA, USA, 1995. Morgan Kaufmann Publishers Inc.
\newblock \doi{10.5555/1643031.1643047}.
\newblock \url{https://doi.org/10.5555/1643031.1643047}.

\bibitem[Hastie et~al.(2009)Hastie, Tibshirani, and
  Friedman]{hastie2009elements}
T.~Hastie, R.~Tibshirani, and J.~Friedman.
\newblock \emph{The Elements of Statistical Learning: Data Mining, Inference,
  and Prediction, Second Edition}.
\newblock Springer Series in Statistics. Springer, New York, NY, 2009.
\newblock \doi{10.1007/978-0-387-84858-7}.

\bibitem[Domingos(2000)]{domingos2000unified}
Pedro~M. Domingos.
\newblock A unified bias-variance decomposition and its applications.
\newblock In Pat Langley, editor, \emph{Proceedings of the Seventeenth
  International Conference on Machine Learning}, pages 231--238, San Francisco,
  CA, 2000. Morgan Kaufmann.
\newblock \doi{10.5555/645529.657784}.
\newblock \url{https://doi.org/10.5555/645529.657784}.

\bibitem[Salton et~al.(1975)Salton, Yang, and Yu]{Salton1975}
G.~Salton, C.~S. Yang, and C.~T. Yu.
\newblock A theory of term importance in automatic text analysis.
\newblock \emph{Journal of the American Society for Information Science},
  26\penalty0 (1):\penalty0 33--44, 1975.
\newblock \doi{10.1002/asi.4630260106}.

\bibitem[Salton and Buckley(1988)]{Salton1988}
Gerard Salton and Christopher Buckley.
\newblock Term-weighting approaches in automatic text retrieval.
\newblock \emph{Information Processing \& Management}, 24\penalty0
  (5):\penalty0 513--523, 1988.
\newblock \doi{10.1016/0306-4573(88)90021-0}.

\bibitem[Torgerson(1952)]{Torgerson1952}
Warren~S. Torgerson.
\newblock Multidimensional scaling: I. theory and method.
\newblock \emph{Psychometrika}, 17\penalty0 (4):\penalty0 401--419, 1952.
\newblock \doi{10.1007/BF02288916}.

\bibitem[Moran(1950)]{Moran1950}
P.~A.~P. Moran.
\newblock Notes on continuous stochastic phenomena.
\newblock \emph{Biometrika}, 37\penalty0 (1-2):\penalty0 17--23, 1950.
\newblock \doi{10.1093/biomet/37.1-2.17}.

\bibitem[Pedregosa et~al.(2011)Pedregosa, Varoquaux, Gramfort, Michel, Thirion,
  Grisel, Blondel, Prettenhofer, Weiss, Dubourg, Vanderplas, Passos,
  Cournapeau, Brucher, Perrot, and Duchesnay]{scikit-learn}
F.~Pedregosa, G.~Varoquaux, A.~Gramfort, V.~Michel, B.~Thirion, O.~Grisel,
  M.~Blondel, P.~Prettenhofer, R.~Weiss, V.~Dubourg, J.~Vanderplas, A.~Passos,
  D.~Cournapeau, M.~Brucher, M.~Perrot, and E.~Duchesnay.
\newblock Scikit-learn: machine learning in {P}ython.
\newblock \emph{Journal of Machine Learning Research}, 12:\penalty0 2825--2830,
  2011.
\newblock \doi{10.5555/1953048.2078195}.

\bibitem[Kruskal(1964)]{kruskal1964}
J.~B. Kruskal.
\newblock Multidimensional scaling by optimizing goodness of fit to a nonmetric
  hypothesis.
\newblock \emph{Psychometrika}, 29\penalty0 (1):\penalty0 1--27, 1964.
\newblock \doi{10.1007/BF02289565}.

\bibitem[Friedman(1997)]{friedman1997bias}
Jerome Friedman.
\newblock On bias, variance, 0/1—loss, and the curse-of-dimensionality.
\newblock \emph{Data Mining and Knowledge Discovery}, 1\penalty0 (1):\penalty0
  55--77, 03 1997.
\newblock \doi{10.1023/A:1009778005914}.

\bibitem[Morris et~al.(2019)Morris, White, and Crowther]{morris2019using}
Tim~P. Morris, Ian~R. White, and Michael~J. Crowther.
\newblock Using simulation studies to evaluate statistical methods.
\newblock \emph{Statistics in Medicine}, 38\penalty0 (11):\penalty0 2074--2102,
  2019.
\newblock \doi{10.1002/sim.8086}.

\bibitem[Efron(1981)]{Efron1981}
Bradley Efron.
\newblock Nonparametric estimates of standard error: The jackknife, the
  bootstrap and other methods.
\newblock \emph{Biometrika}, 68\penalty0 (3):\penalty0 589--599, 1981.
\newblock \doi{10.1093/biomet/68.3.589}.

\bibitem[Rousseeuw(1987)]{rousseeuw1987}
Peter~J. Rousseeuw.
\newblock Silhouettes: a graphical aid to the interpretation and validation of
  cluster analysis.
\newblock \emph{Journal of computational and applied mathematics}, 20:\penalty0
  53--65, 1987.
\newblock \doi{10.1016/0377-0427(87)90125-7}.

\bibitem[Robinson and Foulds(1981)]{Robinson1981}
D.F. Robinson and L.R. Foulds.
\newblock Comparison of phylogenetic trees.
\newblock \emph{Mathematical Biosciences}, 53\penalty0 (1):\penalty0 131--147,
  1981.
\newblock \doi{10.1016/0025-5564(81)90043-2}.

\bibitem[Kelly et~al.(2024)Kelly, Longjohn, and Nottingham]{kelly2024uci}
Markelle Kelly, Rachel Longjohn, and Kolby Nottingham.
\newblock {The UCI Machine Learning Repository}.
\newblock \url{https://archive.ics.uci.edu}, 2024.

\bibitem[German(1987)]{glass_identification}
B.~German.
\newblock Glass identification, 1987.
\newblock \url{https://doi.org/10.24432/C5WW2P}.

\bibitem[Buscema et~al.(2010)Buscema, Terzi, and Tastle]{steel_plates_faults}
M~Buscema, S~Terzi, and W~Tastle.
\newblock {Steel Plates Faults}.
\newblock UCI Machine Learning Repository, 2010.
\newblock \url{https://doi.org/10.24432/C5J88N}.

\bibitem[Higuera et~al.(2015)Higuera, Gardiner, and Cios]{mice_protein}
Clara Higuera, Katheleen Gardiner, and Krzysztof Cios.
\newblock {Mice Protein Expression}.
\newblock UCI Machine Learning Repository, 2015.
\newblock \url{https://doi.org/10.24432/C50S3Z}.

\bibitem[Johnson(2013)]{urban_land_cover}
Brian Johnson.
\newblock {Urban Land Cover}.
\newblock UCI Machine Learning Repository, 2013.
\newblock \url{https://doi.org/10.24432/C53S48}.

\bibitem[Alpaydin and
  Alimoglu(1996)]{pen-based_recognition_of_handwritten_digits_81}
E.~Alpaydin and Fevzi Alimoglu.
\newblock {Pen-Based Recognition of Handwritten Digits}.
\newblock UCI Machine Learning Repository, 1996.
\newblock \url{https://doi.org/10.24432/C5MG6K}.

\bibitem[Michalski and Chilausky(1980)]{soybean_large}
R.S. Michalski and R.L. Chilausky.
\newblock {Soybean (Large)}.
\newblock UCI Machine Learning Repository, 1980.
\newblock \url{https://doi.org/10.24432/C5JG6Z}.

\bibitem[Chipman et~al.(1998)Chipman, George, and McCulloch]{Chipman1998}
Hugh~A. Chipman, Edward~I. George, and Robert~E. McCulloch.
\newblock Bayesian {CART} model search.
\newblock \emph{Journal of the American Statistical Association}, 93\penalty0
  (443):\penalty0 935--948, 1998.
\newblock \doi{10.1080/01621459.1998.10473750}.

\end{thebibliography}

\clearpage
\appendix
\section{Bias and variance proof}\label{appendix:bias_variance}
\begin{proposition}
Fix a class \(m\) and an input point \(x\). Let \(p_m(x)\) denote the true probability of class \(m\) at \(x\), and let \(\hat p_m(x)\) denote the fitted probability for class \(m\) produced by a model trained on a random training set. Then the expected squared error for class \(m\) decomposes into a squared bias term and a variance term
\begin{equation}
\mathbb E\!\left[\left(\hat p_m(x)-p_m(x)\right)^2\right]
=
\left(\mathbb E[\hat p_m(x)]-p_m(x)\right)^2
+
\mathbb E\!\left[\left(\hat p_m(x)-\mathbb E[\hat p_m(x)]\right)^2\right].
\end{equation}
\end{proposition}

\begin{proof}

To decompose the error, add and subtract the mean fitted probability \(\mathbb E[\hat p_m(x)]\):
\[
\hat p_m(x)-p_m(x)
=
\bigl(\hat p_m(x)-\mathbb E[\hat p_m(x)]\bigr)
+
\bigl(\mathbb E[\hat p_m(x)]-p_m(x)\bigr).
\]
Squaring both sides gives
\begin{align}
\left(\hat p_m(x)-p_m(x)\right)^2
&=
\left(\hat p_m(x)-\mathbb E[\hat p_m(x)]\right)^2
\notag\\
&\quad+
2\bigl(\hat p_m(x)-\mathbb E[\hat p_m(x)]\bigr)
\bigl(\mathbb E[\hat p_m(x)]-p_m(x)\bigr)
\notag\\
&\quad+
\left(\mathbb E[\hat p_m(x)]-p_m(x)\right)^2.
\end{align}
Now take expectation:
\begin{align}
\mathbb E\!\left[\left(\hat p_m(x)-p_m(x)\right)^2\right]
&=
\mathbb E\!\left[
\left(\hat p_m(x)-\mathbb E[\hat p_m(x)]\right)^2
\right]
\notag\\
&\quad+
2\bigl(\mathbb E[\hat p_m(x)]-p_m(x)\bigr)
\mathbb E\!\left[
\hat p_m(x)-\mathbb E[\hat p_m(x)]
\right]
\notag\\
&\quad+
\left(\mathbb E[\hat p_m(x)]-p_m(x)\right)^2.
\end{align}
The middle term is zero because
\[
\mathbb E\!\left[
\hat p_m(x)-\mathbb E[\hat p_m(x)]
\right]
=
\mathbb E[\hat p_m(x)]-\mathbb E[\hat p_m(x)]
=
0.
\]
Therefore,
\[
\mathbb E\!\left[\left(\hat p_m(x)-p_m(x)\right)^2\right]
=
\left(\mathbb E[\hat p_m(x)]-p_m(x)\right)^2
+
\mathbb E\!\left[\left(\hat p_m(x)-\mathbb E[\hat p_m(x)]\right)^2\right],
\]
which proves the result.
\end{proof}
\FloatBarrier
\clearpage
\section{Monte Carlo error diagnostics}\label{appendix:mc_error}

Figure~\ref{fig:6-class_repitition} shows the Monte Carlo error for square probability loss, squared bias, and variance across different numbers of simulation repetitions in the 6-class setting. Figure~\ref{fig:7-class_repitition} reports the corresponding results for the 7-class setting.

\begin{figure}[!htbp]
    \centering
    \includegraphics[width=1.05\linewidth]{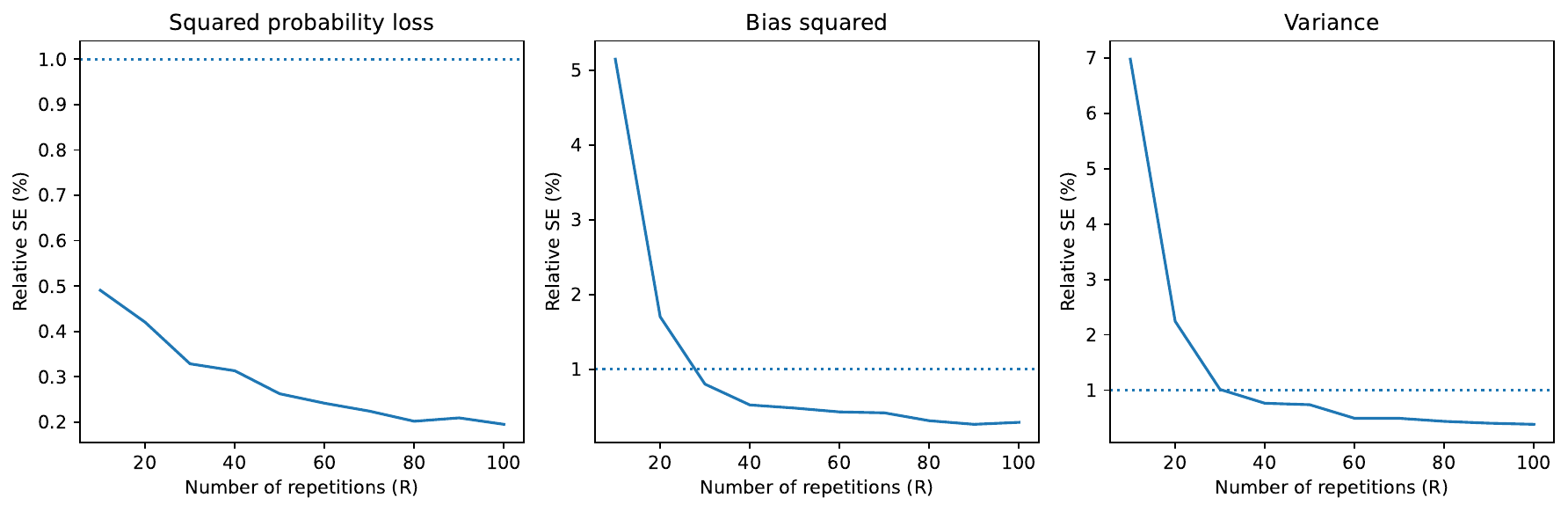}
    \caption{Relative Monte Carlo error against the number of repetitions \(R\) for a 6-class pilot dataset. The panels show squared probability loss, squared bias, and variance. The horizontal dotted line marks the 1\% threshold}
    \label{fig:6-class_repitition}
\end{figure}

\begin{figure}[!htbp]
    \centering
    \includegraphics[width=1.05\linewidth]{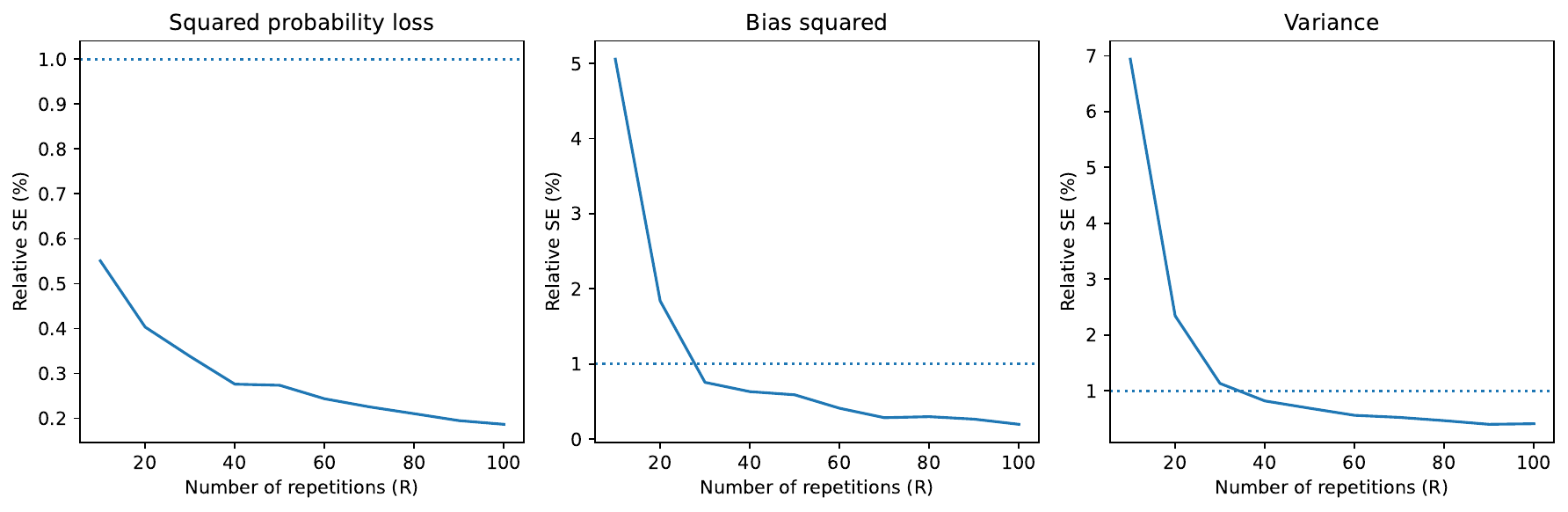}
    \caption{Relative Monte Carlo error against the number of repetitions \(R\) for a 7-class pilot dataset. The panels show squared probability loss, squared bias, and variance. The horizontal dotted line marks the 1\% threshold}
    \label{fig:7-class_repitition}
\end{figure}

\FloatBarrier
\clearpage
\section{Sampling for more than seven classes}\label{appendix:sampling}
The sampler builds each ND top-down. At every internal node, it allocates the remaining budget across possible split sizes in proportion to the number of structures of that size that can be formed from the remaining classes, then recurses on the two children. The allocation is proportional to $\binom{m}{k} W(k) W(m - k)$ (with a factor of $1/2$ when $k = m/2$), where $W(\cdot)$ is the count of rooted binary trees on a given number of leaves. To leading order, this approximates uniform sampling over rooted binary trees and is therefore not biased toward balanced upper-level splits, with deviations introduced only by integer-budget allocation at small budgets.

\section{Decision tree simulation}\label{appendix:dt}
\subsection{The data-generating ND}\label{appendix:target_gen}

Each simulation fixes one generating ND over the classes ${1,\dots,M}$. The ND structure defines class splits, with each internal split being an instance of a binary model. In the decision-tree simulations, this model is a probabilistic CART model. As in a classification tree with terminal-node response parameters \citep{Chipman1998}, a sample is routed through threshold splits to a terminal node, and that terminal node supplies a Bernoulli probability.

For an internal ND node $v$, let $\ell_v(x)$ be the terminal node reached by $x$ in its split tree. The probability of taking the right branch at $v$ is
\[
q_v(x)=p_{v,\ell_v(x)}.
\]
The left branch receives probability $1-q_v(x)$. The class probabilities $p_m(x)$ are then obtained by the same ND recursion used for the logistic data-generating process, multiplying the branch probabilities along the path to class $m$.

Each split tree is fixed when the generating ND is created. Its internal nodes store feature-threshold tests, and its leaves store fixed probabilities
\[
p_{v,\ell}=
\frac{1}{1+e^{-(\beta_v+u_{v,\ell})}},
\qquad
u_{v,\ell}\sim\mathcal{N}(0,\sigma^2).
\]
The intercept $\beta_v$ is calibrated so that the expected right-branch probability matches the relative size of the right class set at that ND split. This targets a roughly uniform marginal class distribution. Generators with extreme calibrated intercepts are rejected, and generated datasets are filtered to avoid very rare or dominant realised classes.

The split trees are part of the data-generating mechanism and are held fixed across repetitions. They are separate from the base learner fitted during evaluation. The fitted model only sees the generated training data.

\subsection{Decision tree results}\label{appendix:dt_results}

The main study fits logistic regression at every node. To test whether the spatial organisation depends on the node model, we repeat the 6-class simulation with a decision tree at each node, using data generated by the CART model above, across the three feature settings in Appendix~\ref{appendix:ablation}. The 50 generating ND structures, the embedding, and the $R = 100$ repetitions are unchanged. The decision tree is \texttt{scikit-learn}'s \texttt{DecisionTreeClassifier} with the entropy criterion, \texttt{min\_samples\_leaf = 5}, \texttt{random\_state = 0}, and no maximum depth \citep{scikit-learn}.

The embedding still orders performance (Table~\ref{tab:dt-simulation}). Squared bias, squared probability loss, and variance all have positive Moran's $I$ in every feature setting, between $0.30$ and $0.37$ and are significant at $p \le 0.001$, so neighbouring trees resemble one another in all three.

\begin{table}[!htbp]
\centering
\caption{Spatial diagnostics for the decision-tree setting at 6 classes, across the three feature settings of Appendix~\ref{appendix:feature_gen}, over the same 50 generators as the main study. $I$ is univariate Moran's $I$, $I_{xy}$ is the bivariate Moran's $I$ between variance and depth, and $r^2$ is the squared Pearson correlation between per-ND squared bias and variance. Values are mean $\pm$ SD, and every $I$ is significant at $p \le 0.001$}
\label{tab:dt-simulation}
\small
\renewcommand{\arraystretch}{1.25}
\begin{tabular*}{\textwidth}{@{\extracolsep{\fill}}lcccc@{}}
\toprule
Setting & $I$ (bias$^2$) & $I$ (var) & $I_{xy}$ (var, depth) & $r^2$(bias$^2$, var) \\
\midrule
10-feat (dense)    & \(0.300 \pm 0.135\) & \(0.356 \pm 0.150\) & \(-0.270 \pm 0.133\) & \(0.882 \pm 0.054\) \\
1000-feat (sparse) & \(0.309 \pm 0.148\) & \(0.365 \pm 0.182\) & \(-0.197 \pm 0.160\) & \(0.888 \pm 0.072\) \\
1000-feat (dense)  & \(0.296 \pm 0.145\) & \(0.357 \pm 0.177\) & \(-0.194 \pm 0.189\) & \(0.892 \pm 0.066\) \\
\bottomrule
\end{tabular*}
\end{table}

Variance behaves differently than under logistic regression. It is nearly flat, sitting near $0.64$ in the 1000-feature setting and varying by only about one per cent of that across the embedding, compared with about eight per cent for squared bias. The depth gradient that organises variance under logistic regression is gone, with the variance--depth bivariate Moran's $I$ turning negative (Table~\ref{tab:dt-simulation}, against $+0.589$ for logistic regression), and what little variation remains follows squared bias, the squared correlation between per-ND squared bias and variance being about $0.89$ ($r \approx 0.94$). This follows from the matched design. The data-generating node model is a decision tree, and the fitted node model is also a decision tree, so structures that poorly match the generating ND tend to have higher bias and higher instability. The decision-tree run therefore shows that the embedding still orders performance under a second node model, but that the depth organisation of variance is specific to the logistic setting.

\section{Simulating features}\label{appendix:feature_gen}

In the simulation settings, an informative feature is one that is used by the branch model at an internal node of the generating ND. In the 1000-feature sparse setting, each node uses 50 features. In dense settings, each node uses all available features, either 10 or 1000. Features not used by a node are noise for that node, even if they may be used by other nodes.

For each simulated dataset, we generate \(p\) features, with \(p=10\) or \(p=1000\) depending on the setting. For feature \(k\in\{1,\dots,p\}\), a distribution family \(D_k\) is first sampled from
\(\{\text{Normal},\text{Poisson},\text{Gamma},\text{Laplace},\text{Binomial}\}\)
with probabilities \((0.30,\,0.10,\,0.15,\,0.20,\,0.25)\). Conditional on \(D_k\), values are then drawn i.i.d.\ for all observations from that family using the parameterisation in our simulation code.

The feature-generation mechanism is held fixed within a simulation. Training features are resampled across repetitions, and the evaluation set is generated once and held fixed (see the \texttt{ndscape} repository \citep{ndscape} for the exact implementation and generation details).

\subsection{Ablation study}\label{appendix:ablation}

The main text (Section~\ref{sec:ablation}) reports ablations for the 1000-feature sparse setting. Here we repeat the same comparison for the 10-feature dense and 1000-feature dense settings, using the same 50 generators, embedding procedures, and evaluation metrics. The depth-weighted tokeniser has the highest mean in every row for both 6 and 7 classes, so its advantage over BoW and RF is not tied to the sparse-feature setting used in the main simulation.

\begin{table}[t]
\centering
\caption{Ablation across feature settings at 6 classes. Mean $\pm$ SD over 50 generators. Bold marks the row maximum. The 1000-feature sparse rows reproduce Table~\ref{tab:ablation-1000-sparse}}
\label{tab:ablation-appendix-6}
\small
\renewcommand{\arraystretch}{1.28}
\begin{tabular*}{\textwidth}{@{\extracolsep{\fill}}lccc@{}}
\toprule
Measure & Tokeniser & BoW & RF \\
\midrule
\multicolumn{4}{l}{\textit{10-feat (dense)}} \\
Moran $I$ (bias$^2$) & \textbf{0.288 $\pm$ 0.103} & 0.233 $\pm$ 0.119 & 0.201 $\pm$ 0.115 \\
Moran $I$ (var) & \textbf{0.323 $\pm$ 0.064} & 0.199 $\pm$ 0.066 & 0.103 $\pm$ 0.057 \\
Moran $I$ (SPL) & \textbf{0.288 $\pm$ 0.103} & 0.233 $\pm$ 0.119 & 0.201 $\pm$ 0.115 \\
AP@10 (bias$^2$) & \textbf{0.338 $\pm$ 0.107} & 0.274 $\pm$ 0.111 & 0.234 $\pm$ 0.079 \\
AP@10 (SPL) & \textbf{0.338 $\pm$ 0.107} & 0.274 $\pm$ 0.112 & 0.234 $\pm$ 0.079 \\
\midrule
\multicolumn{4}{l}{\textit{1000-feat (sparse)}} \\
Moran $I$ (bias$^2$) & \textbf{0.356 $\pm$ 0.155} & 0.256 $\pm$ 0.148 & 0.171 $\pm$ 0.163 \\
Moran $I$ (var) & \textbf{0.541 $\pm$ 0.044} & 0.271 $\pm$ 0.048 & 0.033 $\pm$ 0.019 \\
Moran $I$ (SPL) & \textbf{0.300 $\pm$ 0.138} & 0.248 $\pm$ 0.148 & 0.178 $\pm$ 0.168 \\
AP@10 (bias$^2$) & \textbf{0.372 $\pm$ 0.138} & 0.275 $\pm$ 0.119 & 0.265 $\pm$ 0.115 \\
AP@10 (SPL) & \textbf{0.379 $\pm$ 0.124} & 0.294 $\pm$ 0.119 & 0.273 $\pm$ 0.113 \\
\midrule
\multicolumn{4}{l}{\textit{1000-feat (dense)}} \\
Moran $I$ (bias$^2$) & \textbf{0.355 $\pm$ 0.149} & 0.256 $\pm$ 0.144 & 0.173 $\pm$ 0.158 \\
Moran $I$ (var) & \textbf{0.538 $\pm$ 0.042} & 0.270 $\pm$ 0.052 & 0.038 $\pm$ 0.018 \\
Moran $I$ (SPL) & \textbf{0.297 $\pm$ 0.135} & 0.247 $\pm$ 0.146 & 0.183 $\pm$ 0.166 \\
AP@10 (bias$^2$) & \textbf{0.368 $\pm$ 0.134} & 0.274 $\pm$ 0.117 & 0.267 $\pm$ 0.114 \\
AP@10 (SPL) & \textbf{0.371 $\pm$ 0.121} & 0.289 $\pm$ 0.119 & 0.277 $\pm$ 0.112 \\
\bottomrule
\end{tabular*}
\end{table}

\begin{table}[t]
\centering
\caption{Ablation across feature settings at 7 classes. Mean $\pm$ SD over 50 generators. Bold marks the row maximum. The 1000-feature sparse rows reproduce Table~\ref{tab:ablation-1000-sparse}}
\label{tab:ablation-appendix-7}
\small
\renewcommand{\arraystretch}{1.28}
\begin{tabular*}{\textwidth}{@{\extracolsep{\fill}}lccc@{}}
\toprule
Measure & Tokeniser & BoW & RF \\
\midrule
\multicolumn{4}{l}{\textit{10-feat (dense)}} \\
Moran $I$ (bias$^2$) & \textbf{0.393 $\pm$ 0.071} & 0.233 $\pm$ 0.084 & 0.229 $\pm$ 0.079 \\
Moran $I$ (var) & \textbf{0.413 $\pm$ 0.079} & 0.157 $\pm$ 0.032 & 0.188 $\pm$ 0.060 \\
Moran $I$ (SPL) & \textbf{0.392 $\pm$ 0.071} & 0.233 $\pm$ 0.084 & 0.229 $\pm$ 0.079 \\
AP@10 (bias$^2$) & \textbf{0.258 $\pm$ 0.103} & 0.225 $\pm$ 0.077 & 0.165 $\pm$ 0.042 \\
AP@10 (SPL) & \textbf{0.259 $\pm$ 0.103} & 0.225 $\pm$ 0.077 & 0.165 $\pm$ 0.042 \\
\midrule
\multicolumn{4}{l}{\textit{1000-feat (sparse)}} \\
Moran $I$ (bias$^2$) & \textbf{0.473 $\pm$ 0.119} & 0.256 $\pm$ 0.107 & 0.188 $\pm$ 0.120 \\
Moran $I$ (var) & \textbf{0.726 $\pm$ 0.048} & 0.214 $\pm$ 0.034 & 0.045 $\pm$ 0.018 \\
Moran $I$ (SPL) & \textbf{0.378 $\pm$ 0.103} & 0.263 $\pm$ 0.121 & 0.208 $\pm$ 0.121 \\
AP@10 (bias$^2$) & \textbf{0.297 $\pm$ 0.140} & 0.237 $\pm$ 0.078 & 0.179 $\pm$ 0.059 \\
AP@10 (SPL) & \textbf{0.297 $\pm$ 0.119} & 0.245 $\pm$ 0.081 & 0.191 $\pm$ 0.062 \\
\midrule
\multicolumn{4}{l}{\textit{1000-feat (dense)}} \\
Moran $I$ (bias$^2$) & \textbf{0.475 $\pm$ 0.115} & 0.256 $\pm$ 0.106 & 0.182 $\pm$ 0.113 \\
Moran $I$ (var) & \textbf{0.719 $\pm$ 0.048} & 0.214 $\pm$ 0.038 & 0.046 $\pm$ 0.015 \\
Moran $I$ (SPL) & \textbf{0.379 $\pm$ 0.106} & 0.263 $\pm$ 0.122 & 0.203 $\pm$ 0.116 \\
AP@10 (bias$^2$) & \textbf{0.298 $\pm$ 0.134} & 0.237 $\pm$ 0.079 & 0.177 $\pm$ 0.055 \\
AP@10 (SPL) & \textbf{0.299 $\pm$ 0.115} & 0.245 $\pm$ 0.082 & 0.187 $\pm$ 0.054 \\
\bottomrule
\end{tabular*}
\end{table}

\section{High-class-count embeddings}\label{appendix:highclass}

\begin{figure}[!htbp]\centering
\includegraphics[width=\linewidth]{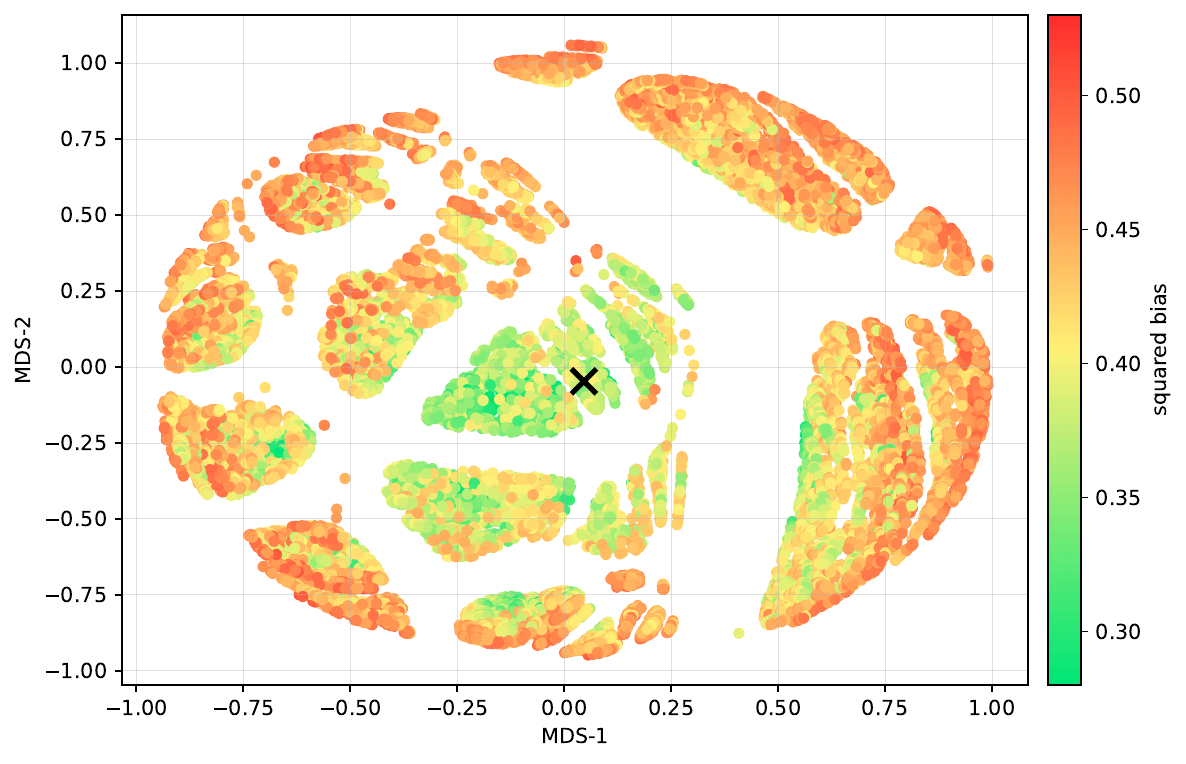}
\caption{Squared bias for 8-class ND embedding}
\label{fig:8-class-bias}
\end{figure}

\begin{figure}[!htbp]\centering
\includegraphics[width=\linewidth]{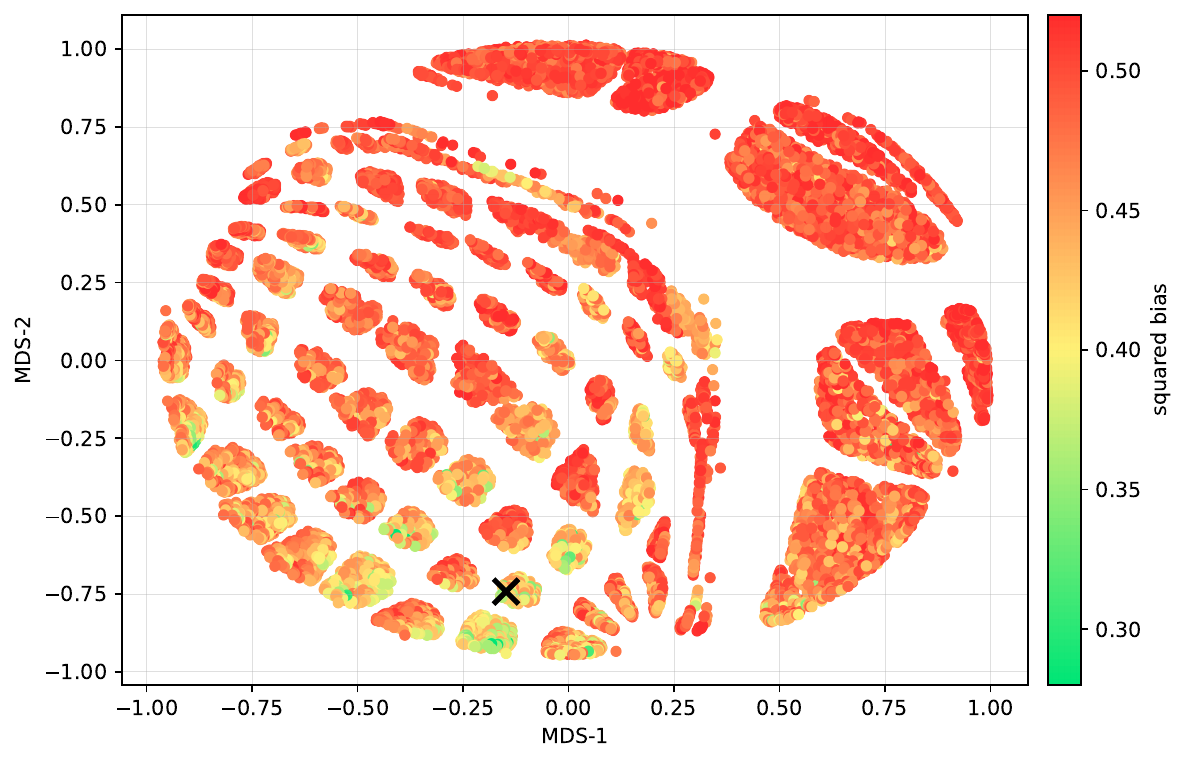}
\caption{Squared bias for 9-class ND embedding}
\label{fig:9-class-bias}
\end{figure}

\begin{figure}[!htbp]\centering
\includegraphics[width=\linewidth]{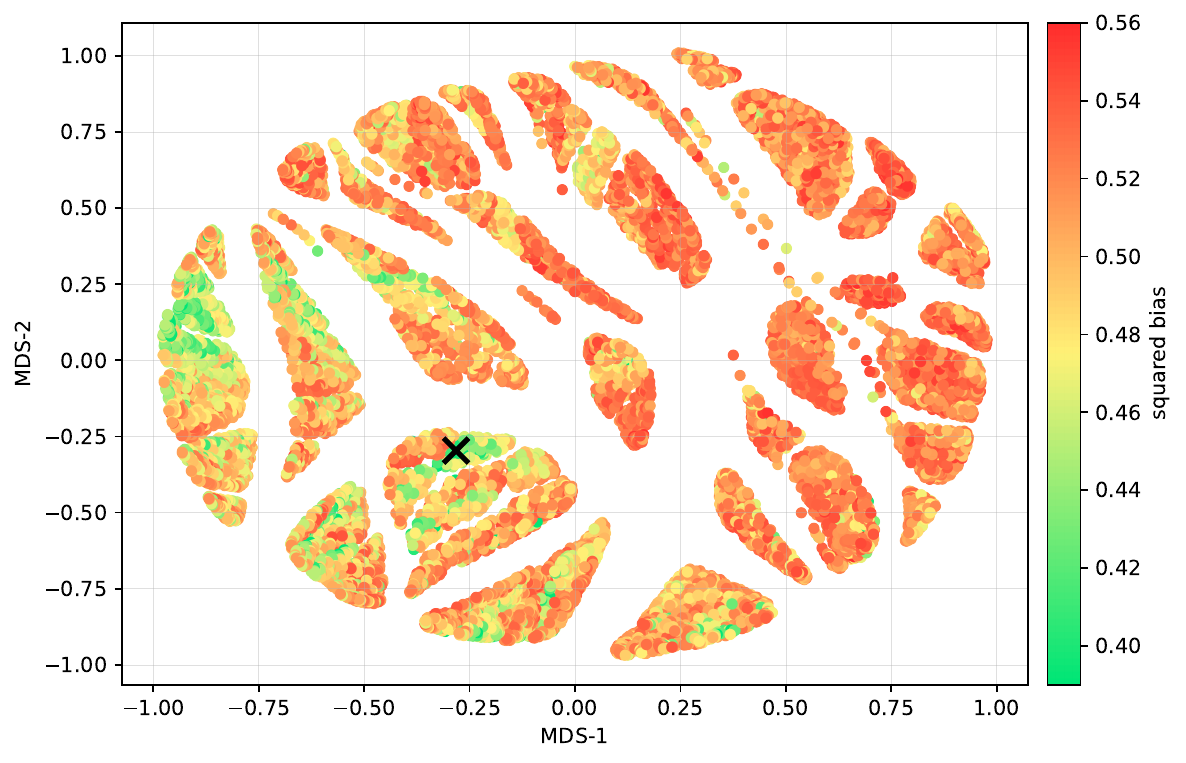}
\caption{Squared bias for 10-class ND embedding}
\label{fig:10-class-bias}
\end{figure}

\FloatBarrier
\clearpage
\section{Per-dataset real-data results}\label{appendix:realdata}

This appendix presents the four datasets not addressed in Section~\ref{sec:real-glass}. The cross-dataset statistics are in Tables~\ref{tab:realdata-spatial} and~\ref{tab:realdata-heuristic}.

\subsection{Steel Plates Faults (7 classes)} 

The dataset has 1{,}941 instances and 27 geometric and luminosity attributes across 7 fault types, and all 10{,}395 NDs are enumerated. Test accuracy ranges from $61.6\%$ to $74.1\%$ (mean $68.4\%$). BoK reaches the top half by cluster rank, the only heuristic on any dataset to do so, landing in a cluster of rank 4 of 10, and its tree at $72.0\%$, tied with NDC, is the best of the six. The other five heuristics are clustered and ranked 7 or worse (Figure~\ref{fig:steel_plates}).

\begin{figure}
    \centering
    \includegraphics[width=1\linewidth]{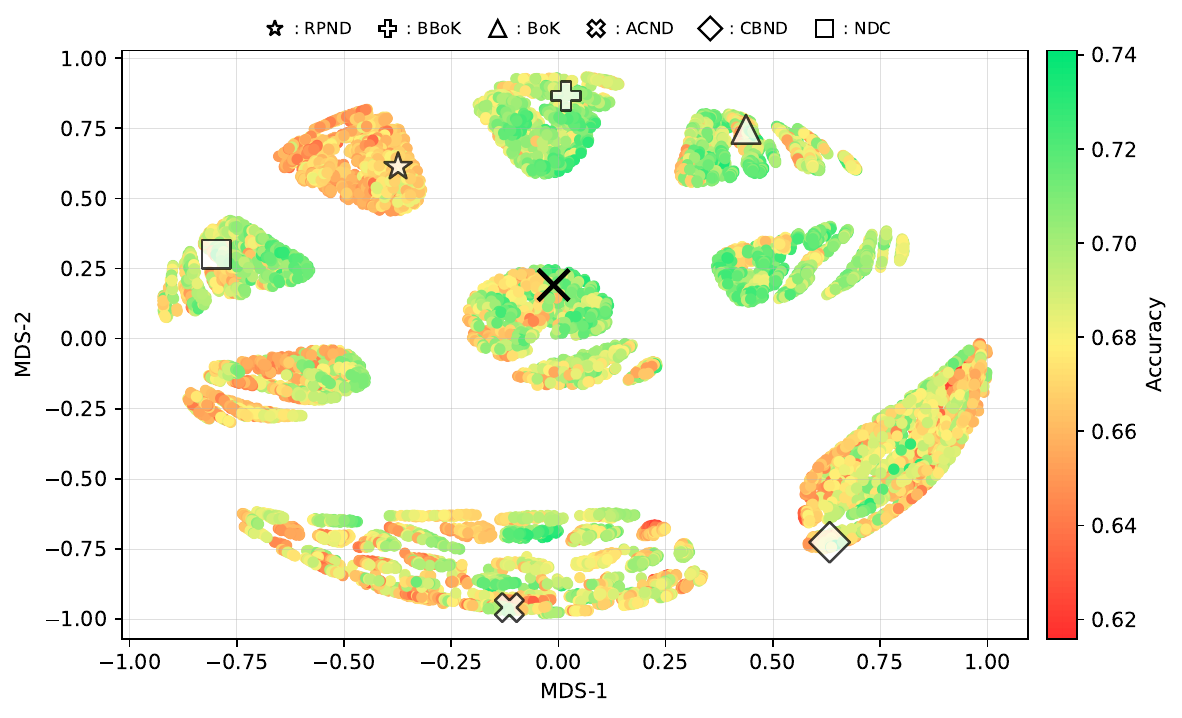}
    \caption{The ndscape embedding for the Steel Plates Faults dataset, coloured by held-out accuracy. Each point is an ND, positioned by MDS on the cosine distance between topology encodings. Marker accuracies are RPND (66.7\%), BoK (72.0\%), BBoK (71.4\%), ACND (71.0\%), CBND (68.6\%), and NDC (72.0\%). The solid cross marks the highest-accuracy ND (74.1\%)}
    \label{fig:steel_plates}
\end{figure}

\FloatBarrier

\subsection{Mice Protein (8 classes)} 

The dataset has 1{,}080 samples and 77 protein expression measurements across 8 classes, sampled to 15{,}000 structures. Test accuracy ranges from $84.9\%$ to $99.7\%$ (mean $96.4\%$). NDC selects the best individual tree at $98.8\%$, within $0.9$ points of the best evaluated tree, yet its cluster has the worst mean of all 19, at $94.1\%$. This is the clearest case of a near-best tree drawn from the weakest region (Figure~\ref{fig:mice_protein}).

\begin{figure}
    \centering
    \includegraphics[width=1\linewidth]{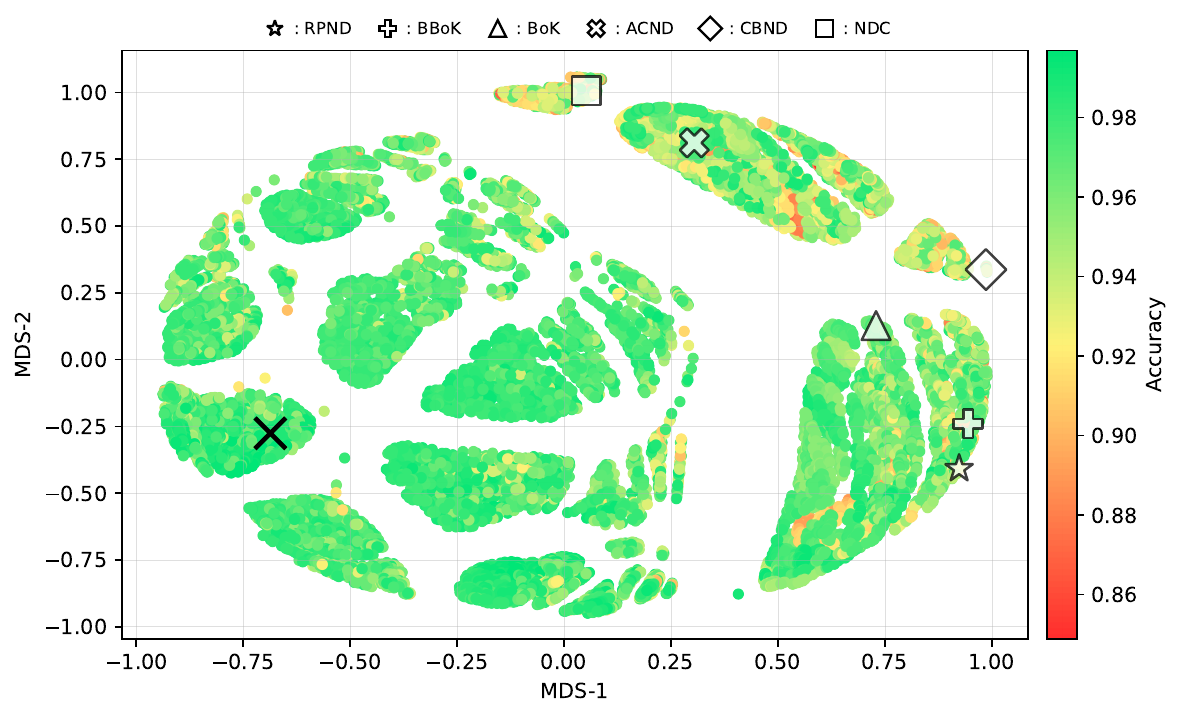}
    \caption{The ndscape embedding for the Mice Protein Expression dataset, coloured by held-out accuracy. Each point is an ND, positioned by MDS on the cosine distance between topology encodings. Marker accuracies are RPND (91.4\%), BoK (97.5\%), BBoK (96.6\%), ACND (94.4\%), CBND (93.5\%), and NDC (98.8\%). The solid cross marks the highest-accuracy ND (99.7\%)}
    \label{fig:mice_protein}
\end{figure}

\subsection{Urban Land Cover (9 classes)}\label{appendix:urban_land} 

The dataset has 675 samples and 147 spectral and morphological features across 9 land-cover classes, with the UCI train-test split and 15{,}000 sampled structures. Test accuracy ranges from $56.6\%$ to $77.3\%$ (mean $69.4\%$). The cluster structure is flat. The 12 clusters span only $3.0$ points from the best mean accuracy to the worst, so cluster-mean accuracy is a weak signal here, which matches the low Moran's $I$ values in Table~\ref{tab:realdata-spatial}. The best heuristic falls into the lowest-ranked cluster, rank 12 of 12 (Figure~\ref{fig:urban_land_cover}). NDC selects the best individual tree at $75.9\%$, within $1.4$ points of the best evaluated tree, yet it sits in the worst-mean cluster. ACND at $73.6\%$ shows the same pattern, a near-best tree drawn from a low-mean region.

\begin{figure}
    \centering
    \includegraphics[width=1\linewidth]{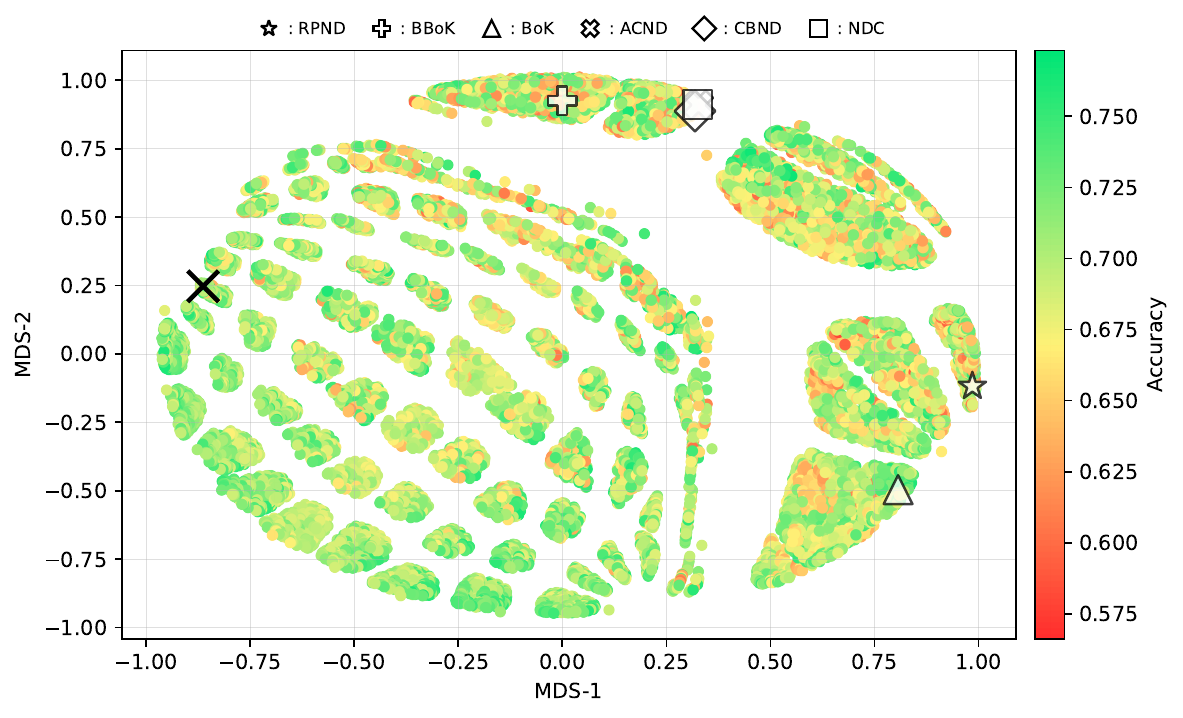}
    \caption{The ndscape embedding for the Urban Land Cover dataset, coloured by held-out accuracy. Each point is an ND, positioned by MDS on the cosine distance between topology encodings. Marker accuracies are RPND (62.9\%), BoK (74.4\%), BBoK (61.1\%), ACND (73.6\%), CBND (64.3\%), and NDC (75.9\%). The solid cross marks the highest-accuracy ND (77.3\%)}
    \label{fig:urban_land_cover}
\end{figure}

\subsection{Pen-Based Digits (10 classes)} 

The dataset has 10{,}992 instances and 16 pen-tip coordinate features across the digits 0 to 9, sampled to 15{,}000 of $3.4 \times 10^7$ structures. Test accuracy ranges from $71.0\%$ to $93.3\%$ (mean $85.9\%$). BoK comes closest to the best region, in a cluster of rank 4 of 17 with a tree at $92.6\%$. The variance autocorrelation is the lowest of the six datasets, which we attribute to the sparse sampling of the largest space (Figure~\ref{fig:pen_based_digits}).

\begin{figure}
    \centering
    \includegraphics[width=1\linewidth]{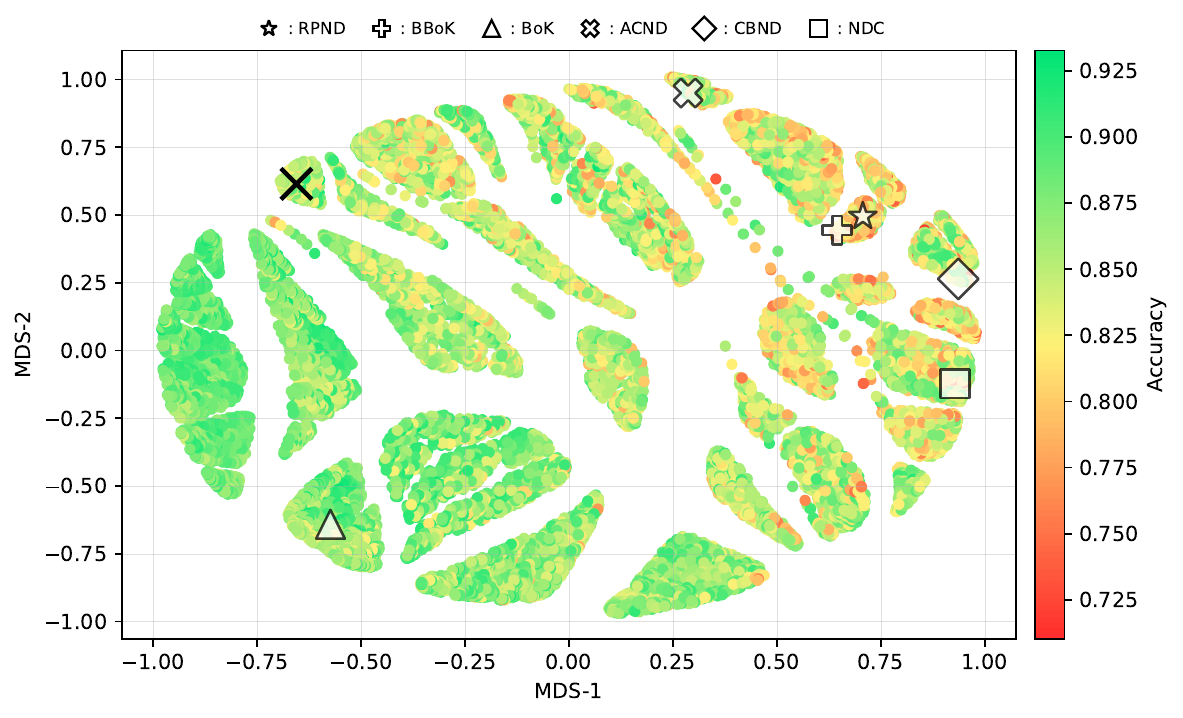}
    \caption{The ndscape embedding for the Pen-Based Recognition of Handwritten Digits dataset, coloured by held-out accuracy. Each point is an ND, positioned by MDS on the cosine distance between topology encodings. Marker accuracies are RPND (82.6\%), BoK (92.6\%), BBoK (90.8\%), ACND (88.8\%), CBND (83.9\%), and NDC (87.7\%). The solid cross marks the highest-accuracy ND (93.3\%)}
    \label{fig:pen_based_digits}
\end{figure}

\subsection{Soybean (19 classes)}\label{appendix:soybean}

The Soybean dataset has 307 samples and 35 features across 19 classes. Here $M - 2 = 17$, so we cap the maximum token depth at $d_{\max} = \min(8, M-2) = 8$. Splits below depth~8 lie in highly unbalanced subtrees and affect a single root-to-leaf path, so the cap bounds the token budget without altering the upper-level structure the map relies on. For every class count up to 10, $M - 2 \le 8$ and the cap has no effect.

We sampled 15{,}000 of \(2.2 \times 10^{20}\) possible structures. Test accuracy ranges from \(86.8\%\) to \(94.6\%\) (mean \(90.7\%\)). Despite the much larger structure space, held-out accuracy, prediction variance, and log loss remain spatially organised in the ndscape embedding (Table~\ref{tab:realdata-spatial}).

NDC selects the best heuristic tree at \(91.2\%\), within \(3.4\) points of the best evaluated tree. However, it lands in a cluster ranked 8 of 27 by mean accuracy. The region gap is \(0.4\) points. As in the other real-data examples, the best individual heuristic tree does not come from the best-performing region of the embedding (Figure~\ref{fig:soybean}).

\begin{figure}
    \centering
    \includegraphics[width=1\linewidth]{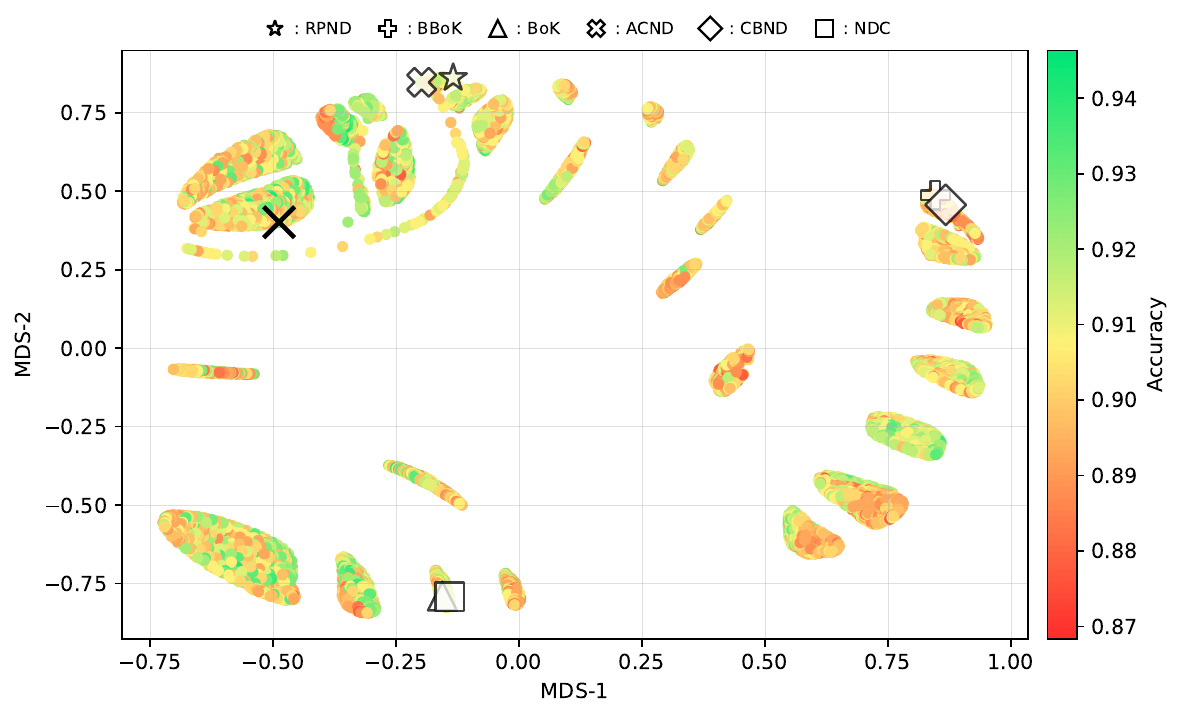}
    \caption{The ndscape embedding for the Soybean dataset, coloured by held-out accuracy. Each point is an ND, positioned by MDS on the cosine distance between topology encodings. Marker accuracies are RPND (\(90.7\%\)), BoK (\(90.2\%\)), BBoK (\(90.7\%\)), ACND (\(90.2\%\)), CBND (\(90.2\%\)), and NDC (\(91.2\%\)). The solid cross marks the highest-accuracy evaluated ND (\(94.6\%\))}
    \label{fig:soybean}
\end{figure}
\FloatBarrier

\end{document}